\newif\ifacs
\acsfalse

\ifacs
\documentclass[sn-mathphys,pdflatex]{sn-jnl}
\else
\documentclass[reqno]{amsart}
\usepackage[foot]{amsaddr}
\fi

\input{macros}

\begin{document}
\ifacs
\input{snFirstPage}
\else
\makeatletter
\newcommand{\definetitlefootnote}[1]{
  \newcommand\addtitlefootnote{
    \makebox[0pt][l]{$^{*}$}
    \footnote{\protect\@titlefootnotetext}
  }
  \newcommand\@titlefootnotetext{\spaceskip=\z@skip $^{*}$#1}
}
\makeatother

\definetitlefootnote{
  Work supported by the ANR project LAMBDACOMB ANR-21-CE48-0017}

\title[Unitless Frobenius quantales]{Unitless Frobenius quantales\addtitlefootnote}
\author[C. De Lacroix]{Cédric De Lacroix}
\email{cedric.delacroix@lis-lab.fr}
\author[L. Santocanale]{Luigi Santocanale}
\email{luigi.santocanale@lis-lab.fr}
\address{
  LIS, CNRS UMR 7020, Aix-Marseille Universit\'e,
  France
}

\maketitle

\begin{abstract}
  It is often stated that Frobenius quantales are necessarily unital.
By taking negation as a primitive operation, we can define Frobenius
quantales that may not have a unit.  We develop the elementary theory
of these structures and show, in particular, how to define nuclei
whose quotients are Frobenius quantales. This yields a phase semantics
and a representation theorem via phase quantales.
\ifacs
\newline
\vspace{-10pt}
\fi

Important examples of these structures arise from Raney's notion of
tight Galois connection: tight endomaps of a complete lattice always
form a Girard quantale which is unital if and only if the lattice is
completely distributive.  We give a characterisation and an
enumeration of tight endomaps of the diamond lattices $M_n$ and
exemplify the Frobenius structure on these maps.  By means of phase
semantics, we exhibit analogous examples built up from trace class
operators on an infinite dimensional Hilbert space.
\ifacs
\newline
\vspace{-10pt}
\fi

Finally, we argue that units cannot be properly added to Frobenius
quantales: every possible extention to a unital quantale fails to
preserve negations.
\ifacs
\newline
\vspace{-10pt}
\fi

\end{abstract}

\smallskip
\noindent \textbf{Keywords.} 
Quantale, Frobenius quantale, Girard quantale, unit, dualizing
element, Serre duality, tight map, trace class operator, nuclear map.

\fi

\maketitle

\section*{Introduction}

It is often stated, see for example
\cite{Rosenthal1990a,Rosenthal1990b,EGHK2018}, that a Frobenius
quantale has a unit.
Indeed, as far as these quantales are defined via a dualizing element
(the linear falsity), this element necessarily is the unit of the dual
multiplication. Then, by duality, the standard multiplication of the
quantale has a unit. Negations are here taken as defined operators, by
means of false and the implications.

It is possible, however, to consider negations as primitive operators
and axiomatize them so to be coherent with the implications. 
We follow here this approach, thus exploring an axiomatization which
might be considered folklore: while the axiomatization is explicitly
considered in \cite[\S
3.3]{GJKO}, 
it is also closely related to the notions of \emph{Girard couple}
\cite{EggerKruml2010} and of \emph{Serre duality} \cite{Rump2021}.
Models of the given axioms are quantales coming with a notion of
negation, yet they might lack a unit. When they have a unit, these
structures coincide with the standard Frobenius and Girard
quantales. We call these structures 
\emph{unitless Frobenius quantales}.
It is the goal of this paper to explore in depth such an
axiomatization. We get to the conclusion that unitless Frobenius and
Girard quantales are structures of interest and worth further
research.

In support of this conclusion we present several examples of these
structures and characterize when they have units. Also, we show that
the standard theory of quantic nuclei and phase quantales can be
lifted to unitless Girard quantales and even to unitless Frobenius
quantales.  Slightly surprisingly, the new axiomatization immediately
clarifies how to generalise to Frobenius quantales the usual double
negation nucleus. As far as we are ware of, this nucleus is always
described in the literature as arising from a cyclic element, yielding
therefore a Girard quantale and, overall, a representation theorem for
Girard quantales. 
We prove therefore that unitless Frobenius quantales can be
represented as phase quantales, i.e. as quotients of free quantales
over a semigroup by a double negation nucleus. As a consequence, many
of our examples (and in principle all of them) arise as phase
quantales.

We present and study the example that prompted us to develop this
research. It is known
\cite{KrumlPaseka2008,EggerKruml2010,EGHK2018,S_RAMICS2020,S_ACT2020}
that the quantale of \jp endomaps of $L$ has the structure of a
Frobenius quantale if and only if $L$ is a \cd lattice. If $L$ is \cd,
then every \jp endomap of $L$ is tight, in the sense of
\cite{Raney60}, or nuclear, in the sense of \cite{HiggsRowe1989}.  If
$L$ is not \cd, then tight maps still form a Girard quantale, yet a
unitless one as we define. A statement finer than the one above then
asserts that the Girard quantale of tight endomaps of a complete
lattice $L$ is unital if and only if $L$ is \cd.  We illustrate the
theory so far developed on this example by showing that the Girard
quantale of tight endomaps can also be constructed via a double
negation nucleus.
We further illustrate this example by characterizing tight endomaps of
$M_{n}$ (the modular lattice of height $3$ with $n$ atoms), by
enumerating these maps, by characterising the operations on tight maps
arising from the quantale structure.

In this setting, a natural question is whether it is possible to embed
unitless Frobenius quantales into unital ones. This question has
always a trivial positive answer. Yet, given the choice of negations
as primitive operations, a finer question is whether it is possible to
find such an embedding that also preserves the negations.  A simple
but surprising argument shows that this is never possible, unless the
given quantale has already a unit. Studying further this phenomenon,
we emphasize that, contrary to the unital case, this is due to the
fact that positive elements are not closed under infima.

Let us mention other motivations that prompted us to develop this
research.  The literature on quantales often emphasizes that these
structures might not have units. We could not identify a work whose
main topic are units in quantales, yet the role of units has been discussed
from slightly different perspectives.
For example, completeness and complexity results for the Lambek
calculus---which we consider as a fragment of non-commutative linear
logic whose connections to quantale theory are well-known
\cite{Yetter1990}---have shown to importantly differ depending on the
presence of a unit in the calculus
\cite{AndrekaMikulas1994,Kuznetsov2021}.

The need of considering unitless semigroups analogous to Frobenius
quantales has shown up when devising categorical frameworks for
quantum computation \cite{AbramskyHeunen2012}. While the Frobenius
algebras considered in that work are not quantales, the Frobenius
quantales that we construct in Section~\ref{sec:phaseQuantales} are
built up from the same Frobenius algebras of trace class operators
considered in \cite{AbramskyHeunen2012}; the strict analogy between
the two situations can be formalised using the language of autonomous
categories.
 
Finally, coming to the past research of one of the authors, the
constructions carried out in \cite{SAN-2021-JPAA} rely on Girard
quantales; however, to perform them, it is unclear whether the units
play any role, units even appear to be problematic. Fundamental
results in that paper depend on considering an infinite family of
Girard quantale embeddings that preserve all the structure except for the
units.

\ifacs\else
\tableofcontents
\fi
\section{Elementary notions}

Let us recall a few notions that we need in the course of the
paper. We address the reader to standard monographies, e.g.
\cite{DP2002,Rosenthal1990a,KrumlPaseka2008,GJKO,EGHK2018}, for an in-depth
presentation of these notions.
\begin{definition}
  A \emph{complete lattice} $L$ is a poset such that suprema
  $\bigvee_{i \in I} x_{i}$ exists, for each family
  $\set{x_{i} \in L\mid i \in I}$. A function $f : L_{0} \rto L_{1}$ between two posets
  is \emph{\jp} if $f(\bigvee_{i \in I} x_{i}) =\bigvee_{i \in I} f(x_{i})$,
  for each family $\set{x_{i} \in L_{0}\mid i \in I}$ and whenever the
  supremum  $\bigvee_{i \in I} x_{i}$ exists.
\end{definition}
If $L$ is a complete lattice, then the poset $L^{\op}$ (with the
opposite ordering, $x \leq_{op} y$ if and only if $y \leq x$) is also
a complete lattice, where the suprema in $L^{\op}$ are infima in $L$.
If $g : L_{0} \rto L_{1}$ is \jp, then there exists a unique \jp map
$g : L^{\op}_{1} \rto L_{0}^{\op}$ satisfying $f(x) \leq y$ if and
only if $g(y) \leq_{op} x$---that is, $x \leq g(y)$---for each
$x,y \in L$. We use the notation $\radj{f}$ for this map and say that
$\radj{f}$ is the \emph{right adjoint} of $f$. In a similar way, if
$g = \radj{f}$, then we let $f = \ladj{g}$ and say that $f$ is the
\emph{left adjoint} of $g$. We also write $f \adj g$ if $g = \radj{f}$
and $f = \ladj{g}$.

Complete lattices and \jp maps form a category, denoted by
\SLatt.  The operation $\fun^{\op} : \SLatt \rto \SLatt$, that we can
extend to \jp maps by setting $f^{\op} \eqdef \radj{f}$, is a
contravariant functor from $\SLatt$ to itself. More than that, it is
actually a category isomorphism between $\SLatt$ and its opposite
category $\SLatt^{\op}$.  A map $g : L_{0} \rto L_{1}$ is \emph{\mp} if
it can be considered as a map $L_{0}^{\op} \rto L_{1}^{\op}$ in \SLatt.

A \emph{Galois connection} on complete lattices $L_{0},L_{1}$ is a
pair of maps $f : L_{0} \rto L_{1}$ and $g : L_{1} \rto L_{0}$ such
that, for each $x \in L_{0}$ and $y \in L_{1}$, we have $y \leq f(x)$
if and only if $x \leq g(y)$. It is a consequence of this definition
that we can see $f$ as a \jp map from $L_{0}$ to $L_{1}^{\op}$ and that
$g = \radj{f} : L_{1} = (L_{1}^{\op})^{\op} \rto L_{0}^{\op}$.
Every Galois connection is uniquely determined in this way.
We say that $f: L \rto L$ is
\emph{self-adjoint} if the pair $f,f$ is a Galois connection.

A \emph{closure operator} on a \cl $L$ is a map $j : L \rto L$ which is
isotone (i.e. order preserving), increasing ($x \leq j(x)$, for each
$x \in L$), and idempotent ($j(j(x)) = j(x)$, for each $x \in L$). If
$j : L \rto L$ is a closure operator, then $L_{j} \eqdef \set{x \in L
  \mid j(x) =  x}$ is, with the ordering  induced  from $L$, a \cl,
where suprema (in $L_{j}$, denoted by $\supj$) are computed  as follows:
\begin{align*}
  \supj \set{x_{i} \in L_{j} \mid i \in I}
  & = j(\,\bigvee \set{x_{i} \in L_{j} \mid i \in I}\,)\,.
\end{align*}
We shall say that $x \in L$ is \emph{$j$-closed} if $x \in L_{j}$.  The map
$j : L \rto L_{j}$ is then a surjective \jp map.  If
$f : L_{0} \rto L_{1}$ is \jp, then $\radj{f}\circ f$ is a closure
operator in $L_{0}$, and every closure operator can be obtained in
this way.

\begin{definition}
  A \emph{quantale} is a pair $(Q,\qmult)$ where $Q$ is a complete lattice
  and $\qmult$ is a semigroup operation that distributes with
  arbitrary suprema in each place:
  \begin{align}
    \label{eq:multdistr}
    (\,\bigvee_{i \in I} x_{i}\,) \qmult (\,\bigvee_{j \in J} y_{j}\,)
    & = \bigvee_{i \in I, j \in J} x_{i} \qmult y_{j}\,,
  \end{align}
  for each pair of families $\set{x_{i} \mid i \in I}$ and
  $\set{y_{j} \mid j \in J}$.  If the semigroup operation $\qmult$ has
  a unit, then we say that the quantale is \emph{unital}. If
  $(Q_{i},\qmult_{i})$, $i = 0,1$, are quantales, then a \jp map
  $f : Q_{0} \rto Q_{1}$ is a \emph{quantale homomorphism} if
  $f(x \qmult_{0} y) = f(x) \qmult_{1} f(y)$, for each
  $x,y \in Q_{0}$.
\end{definition}
The category \SLatt is actually a \sautcat, see
e.g. \cite{Barr1979,HiggsRowe1989}. As such, it comes with a tensor
product and 
a (unital) quantale is exactly a (monoid) semigroup object in the
monoidal category \SLatt. For a quantale $(Q,\qmult)$ and fixed
$x,y \in Q$, equation~\eqref{eq:multdistr} exhibits the maps
$x\qmult \fun : Q \rto Q$ and $\fun \qmult y : Q \rto Q$ as being
\jp. Thus thes two maps have right adjoints
$x \lrimpl \fun : Q^{\op} \rto Q^{\op}$ and
$\fun \rlimpl y : Q^{\op} \rto Q^{\op}$. We have therefore
\begin{align*}
  x \qmult y \leq z & \Tiff y \leq x \lrimpl z \Tiff x \leq z \rlimpl y\,,
\end{align*}
for each $x,y,z \in Q$. These operations, that we call the
\emph{implications}, have also the following explicit expressions:
\begin{align*}
  x \lrimpl z & = \bigvee \set{y \mid x \qmult y  \leq z}\,,
  \qquad z \rlimpl y = \bigvee \set{x \mid x \qmult y  \leq z} \,.
\end{align*}
Notice that if $f : Q_{0} \rto Q_{1}$ is a quantale homomorphism, then
$f(x \lrimpl z) \leq f(x) \lrimpl f(z)$ but in general this inequality
might be strict (and a similar remark applies to $f(z \rlimpl y)$).

\section{Frobenius and Girard quantales}

The usual definition of Frobenius and Girard quantales requires a
dualizing possibliy cyclic element. It goes as follows:
\begin{definition}[Definition A]
  \label{defi_A}
  Let $(Q,\qmult)$ be a quantale and let $0 \in Q$.  The element $0$ is
  \emph{dualizing} if, for every $x$ in $Q$, we have
  \begin{align*}
    0 \rlimpl (x \lrimpl 0) & = (0 \rlimpl x) \lrimpl 0 = x\,. 
  \end{align*}
  The element $0$ is \emph{cyclic} if, for every $x$ in $Q$, we have
  \begin{align*}
    x \lrimpl 0 & = 0 \rlimpl x\,. 
  \end{align*}
  A \emph{Frobenius quantale} is a tuple $(Q,\qmult, 0)$ where
  $(Q,\qmult)$ is a quantale and $0 \in Q$ is dualizing. If moreover $0$
  is cyclic then $(Q,\qmult, 0)$ is a \emph{Girard quantale}.
\end{definition}
 
It is well known that a Frobenius quantale, if so defined, is unital,
see \cite[Corollary to Proposition 1]{Rosenthal1990b} and
\cite[Proposition 2.6.3]{EGHK2018}.  Indeed,
$0\lrimpl 0 = 0 \rlimpl 0$ turns out to be the unit of the quantale.
We propose next a slightly different definition of
Frobenius and Girard quantales, which does not rely on dualizing
elements. The definition is more general in that, as we shall see, it
allows to consider unitless quantales.

\begin{definition}[Definition B]
    \label{defi_B}
    A \emph{Frobenius quantale} is a tuple
    $(Q,\qmult,\Lneg{(-)}, \Rneg{(-)})$ where $(Q,\qmult)$ is a quantale and
    $\Lneg{(-)}, \Rneg{(-)}:Q\rto Q$ are inverse antitone maps
    satisfying
    \begin{align}
      \label{eq:Serre}
      x \lrimpl \Lneg{y} & = \Rneg{x} \rlimpl y\,, \qquad\text{for every $x,y \in Q$.}
    \end{align}
    The map $\Rneg{(-)}$ is called the \emph{right negation} while the
    map $\Lneg{(-)}$ the \emph{left negation}.  A \emph{Girard
      quantale} is a Frobenius quantale for which right and left
    negations coincide.
  \end{definition}  
  We informally call a structure as defined in Definition~\ref{defi_B}
  a \emph{unitless Frobenius} (or \emph{Girard}) \emph{quantale}. This
  naming, however, is slight imprecise, as these quantales might well
  have a unit.

  Let us make straight the relation between  Definition~\ref{defi_A}
  and Definition~\ref{defi_B}.
  \def\thiscite{\cite[Lemma 3.18]{GJKO}}
  \begin{lemma}[cf. \thiscite]
    The Frobenius quantale structures defined in
    Definition~\ref{defi_A} and those defined in
    Definition~\ref{defi_B} that moreover are unital, are in
    bijection.
  \end{lemma}
  \begin{proof}
    Let $(Q,\qmult,0)$ be a Frobenius quantale as defined in
    Definition~\ref{defi_A}.  By letting $\Rneg{x} := x\lrimpl 0$ and
    $\Lneg{x} := 0\rlimpl x$ one obtains a Frobenius quantale as
    defined in Definition~\ref{defi_B}, in view of the well known
    identity $x \lrimpl (z \rlimpl y) = (x \lrimpl z) \rlimpl y$.

    In the opposite direction, let $(Q,\qmult,\Lneg{(-)},\Rneg{(-)})$ be
    a Frobenius quantale as defined in Definition~\ref{defi_B} which,
    moreover, is unital. Observe that
    $\Lneg{1} = 1 \lrimpl \Lneg{1} = \Rneg{1}\rlimpl 1 =
    \Rneg{1}$. Thus, let $0 \eqdef \Lneg{1} = \Rneg{1}$ and observe
    that
    \begin{align*}
      \Lneg{x} & = 1 \lrimpl \Lneg{x} = \Rneg{1} \rlimpl x = 0 \rlimpl
      x
    \end{align*}
    and, similarly, $\Rneg{x} = x \lrimpl 0$. It is then immediate
    that $0$ is a dualizing element.
  \end{proof}

  The identity~\eqref{eq:Serre} is considered in \cite{GJKO}, where it
  is called \emph{contraposition}.  Let us consider a few consequences
  of~\eqref{eq:Serre}. This identity is (under the assumption that
  negations are inverse to each other) clearly equivalent to the
  following ones:
  \begin{align}
    \label{eqs:Serre}
    x \lrimpl y & =
    \Rneg{x} \rlimpl \Rneg{y}\,, & x \rlimpl y & = \Lneg{x} \lrimpl
    \Lneg{y}\,,
    &
    \Lneg{x} \lrimpl y & = x \rlimpl \Rneg{y}\,.
  \end{align}
  Invertible maps satisfying the above kind of conditions were called
  Serre dualities in \cite{Rump2021}. The last of the three identities
  above was considered in \cite{EggerKruml2010}.
  \begin{definition}
    We say that a pair of antitone maps
    $\Lneg{(-)}, \Rneg{(-)}:Q\rto Q$ is \emph{Serre} if they
    satisfy~\eqref{eq:Serre}.  
    If they are inverse to
    each other, then we say that they are form a \emph{Serre duality}.
  \end{definition}
  The following Lemma immediately follows from the previous
  definition.
  \begin{lemma}
    A pair of antitone maps is Serre if and only if the equivalence
    below holds (for each $x,y,z \in Q$):
    \begin{align}
      \label{eq:shift}
      x \qmult z & \leq \Lneg{y} \Tiff z \qmult y \leq \Rneg{x}\,. 
    \end{align}
  \end{lemma}
  Notice that 
  the equivalence in~\eqref{eq:shift}
  amounts to what we called  elsewhere (see e.g. \cite{S_RAMICS2021})
  the \emph{shift relations}:
  \begin{align*}
    x \qmult y \leq z & \Tiff \Lneg{z} \qmult x \leq \Lneg{y} \Tiff y \qmult
    \Rneg{z} \leq \Rneg{x}\,,
  \end{align*}
  which also appear in \cite[\S 3.2.2]{GJKO} under the naming of
  (IGP).
  
  \medskip
  
  If $Q$ is unital, then a Serre pair $\Lneg{\fun}, \Rneg{\fun}$
  necessarily form a Galois connection.  If $Q$ is not unital, then a
  Serre pair need not to be a Galois connection. For example, if
  $x \qmult y = 0$ for each $x,y \in Q$, then the constant map
  $\Rneg{x} = 0$, $x \in Q$, is Serre with itself, but it is not
  self-adjoint if $Q$ has at least two elements.  We shall restrict
  ourself, notably in the next section, to consider Serre pairs
  forming a Galois connection. Of course, if $\Lneg{\fun},\Rneg{\fun}$
  are inverse to each other, then they are also adjoint.

  \medskip
  
  Given a quantale $(Q,\qmult)$ and pair of antitone maps
  $\Lneg{(-)}, \Rneg{(-)}:Q\rto Q$ inverse to each other, we can
  define two distinct dual operators:
  \begin{align*}
    x \Roplus y & \eqdef \Lneg{(\Rneg{y} \qmult \Rneg{x})}\,, &
    x \Loplus
    y & \eqdef \Rneg{(\Lneg{y} \qmult \Lneg{x})}\,.
  \end{align*}
  \def\thiscite{\cite[Lemma 3.17]{GJKO}}
  \begin{proposition}[cf. \thiscite]
    If $\Lneg{\fun},\Rneg{\fun}$ is a Serre duality, then the two dual
    multiplications coincide and they are determined by the
    implication and the negations as follows
    \begin{align*}
      x \Roplus y & =  \Lneg{x} \lrimpl y = x \rlimpl \Rneg{y} = x
      \Loplus y\,.
    \end{align*}
  \end{proposition}
  \begin{proof}
    The chain of equivalences
    \begin{align*}
      z & \leq \Rneg{(\Lneg{y} \qmult \Lneg{x})}
      \Tiff \Lneg{y} \qmult
      \Lneg{x} \leq \Lneg{z}
      \Tiff
      \Lneg{x} \qmult z \leq y
      \Tiff
    z \leq \Lneg{x} \lrimpl y
  \end{align*}
  yields the identity $x \Loplus y = \Lneg{x} \lrimpl y$.  In a
  similar way, we have
  \begin{align*}
    z & \leq \Lneg{(\Rneg{y} \qmult \Rneg{x})}
    \Tiff \Rneg{y} \qmult
    \Rneg{x} \leq \Rneg{z}
    \Tiff
    z \qmult \Rneg{y} \leq x
    \Tiff
    z \leq x \rlimpl \Rneg{y}
  \end{align*}
  yielding $x \Roplus y = x \rlimpl \Rneg{y}$.  The equality
  $x \Roplus y = x \Loplus y$ then follows from~\eqref{eqs:Serre}.
\end{proof}

Our next aim is that of providing simple examples of unitless Frobenius
quantales that indeed lack units.

\begin{Example}[Chu construction]
  We give a first example of unitless Girard quantale by observing
  that the usual Chu construction (or Twist product, see
  e.g. \cite{Kalman1958,Barr1979,Rosenthal1990b}), does not require a quantale to
  be unital and yields a Girard quantale as defined in
  Definition~\ref{defi_B}.
  For a quantale $(Q,\ast)$, the quantale $C(Q)=(Q\times\Qd, \star)$
  is defined by
  \begin{align*}
    (x_1,x_2)\star (y_1,y_2) & \eqdef (x_1\ast y_1, y_1\lrimpl
    x_2\wedge y_2\rlimpl x_1)\,.
  \end{align*}
One can directly check the associativity of $\star$ and that it
distributes in both arguments over joins of $Q \times
\Qd$. Recall that
\begin{align*}
(x_1,x_2)\lrimpl (z_1,z_2) & = (x_{1}\lrimpl z_{1} \wedge x_2\rlimpl z_2,
 z_2 \ast x_{1})\,,\\
 (z_1,z_2)\rlimpl (y_1,y_2) & = (z_1\rlimpl y_1 \wedge z_2\lrimpl y_2,
 y_1\ast z_2)\,,
\end{align*}
since the conditions
$(x_{1},x_{2}) \star (y_{1},y_{2}) \leq (z_{1},z_{2})$,
$(y_{1},y_{2}) \leq (x_{1},x_{2}) \lrimpl (z_{1},z_{2})$, and
$(x_{1},x_{2}) \leq (z_{1},z_{2}) \rlimpl (y_{1},y_{2})$, are all
equivalent to the conjuction of the three conditions:
\begin{align*}
  x_{1} \ast y_{1} & \leq z_{1}\,, \quad y_{1} \ast z_{2} \leq
  x_{2}\,, \quad z_{2} \ast x_{1} \leq y_{2}\,.
\end{align*}
The duality being given by
\begin{align*}
\Rneg{(x_1,x_2)} &\eqdef (x_2,x_1)\,,
\intertext{so the computations}
(x_{1},x_{2})\lrimpl \Rneg{(y_{1},y_{2})} & = (x_{1},x_{2})\lrimpl
(y_{2},y_{1}) = (x_{1}\lrimpl y_{2} \wedge x_{2}\rlimpl y_{1}, y_{1} \ast x_{1}) \\
\Rneg{(x_1,x_2)}\rlimpl (y_1,y_2) & =(x_2,x_1)\rlimpl (y_1,y_2) =
(x_2\rlimpl y_1 \wedge x_1\lrimpl y_2, y_1\ast x_1)
\end{align*}
exhibit $\Rneg{\fun}$ as Serre self-dual.

\begin{proposition}
  The quantale $C(Q)$ is unital if and only if $Q$ is unital.
\end{proposition}
Indeed, the first projection is a surjective semigroup homomorphism and
so, if $C(Q)$ has a unit $(u_{1},u_{2})$, then $u_{1}$ is a unit of
$Q$.  On the other hand, it is well known see that if $u_{1}$ is a
unit of $Q$, then $(u_{1},\top)$ is a unit of $C(Q)$.
\end{Example}

\begin{Example}[Couples of quantales]
  In \cite{EggerKruml2010} the authors define a couple of quantales as
  a pair of quantales $C,Q$ that are related by a sup-preserving map
  $\phi : C \rto Q$; moreover, $C$ is asked to be a $Q$-bimodule and
  the folllowing equations
  \begin{align}
    \label{eq:Gcouple}
    \phi(c_{1}) \cdot c_{2} &= c_{1}\cdot \phi(c_{2}) = c_{1}\qmult c_{2}
  \end{align}
  are to be satisfied.  If $Q$ is a quantale, then $C = Q^{\op}$ is a
  canonical $Q$-bimodule, with, for $x,y \in Q$ and $q \in Q^{\op}$,
  \begin{align*}
    y \cdot q & = q \rlimpl y\,,  \qquad q \cdot x = x \lrimpl q\,.  
  \end{align*}
  Thinking of $\phi : Q^{\op} \rto Q$ as a sort of negation, the first
  equation in \eqref{eq:Gcouple} yields
  \begin{align*}
    x \rlimpl \phi(y) & = \phi(x) \lrimpl y\,,
  \end{align*}
  that is, the last equation in \eqref{eqs:Serre}.  Notice that, by
  taking the above equality as definition of a binary operator on
  $Q^{\op}$, then this operator is necessarily associative.  If,
  moreover, $\phi$ is an involution, then $Q$ becomes a unitless
  Girard quantale. Viceversa, unitless Girard quantale structures on
  $Q$ give rise to couples of quantales $\phi : C \rto Q$ with
  $C = Q^{\op}$ and $\phi$ an antitone involution. 
  Let us notice, however, that \emph{Girard} couples of quantales, in
  contrast to couples of quantales, turn out to have some unit
  \cite[Proposition 8]{EggerKruml2010}.
\end{Example}

\section{Nuclei and phase quantales}
\label{sec:phaseQuantales}

 In this section we show how to generalise the elementary theory of
 nuclei from unital Girard
quantales  to unitless Girard quantales and also unitless Frobenius
quantales. 

\medskip

Let us recall that a \emph{quantic nucleus} (or simply a
\emph{nuclues}) on a quantale $(Q,\qmult)$ is a closure operator $j$
on $Q$ such that, for all $x,y \in Q$,
\begin{align*}
  j(x) \qmult j(y) & \leq j(x \qmult y)\,.
\end{align*}
It is easily seen that a closure operator $j$ is a nucleus if and only
if the two conditions below hold:
\begin{align*}
  x \qmult j(y) & \leq j (x \qmult y)\,, \qquad j(x) \qmult y \leq j(x \qmult
  y),\qquad\qquad\text{for all $x,y \in Q$}\,.
\end{align*}

Given a nucleus $j$ on $(Q,\qmult)$, let $Q_{j}$ be the set of fixed
points of $j$. From the elementary theory of closure operators,
$Q_{j}$ is a complete lattice and $j : Q \rto Q_{j}$ is a surjective
\jp map. $Q_{j}$ has a canonical structure of a quantale
$(Q_{j}, \qmult_{j})$ as well, where the multiplication is given by
\begin{align*}
  x \qmult_{j} y & \eqdef j(x \qmult y)\,.
\end{align*}
Moreover, $j : (Q,\qmult) \rto (Q_{j},\qmult_{j})$ is a surjective
quantale morphism.

\begin{definition}
  For a quantale $(Q,\qmult)$, a \emph{Serre Galois connection} is a
  Galois connection $l,r : Q \rto Q$ such that $ l \circ r = r \circ l$
  and, for all $x,y,z \in Q$,
  \begin{align}
    x \qmult y & \leq r(z) \Tiff z \qmult x \leq l(y)\,.
    \label{eq:shift}
  \end{align}
\end{definition}
That is, $(l,r)$ is Serre if $l \circ r = r \circ l$ and the following
identity holds:
\begin{align*}
  r(z) \rlimpl y & = z \lrimpl l(y)\,.
\end{align*}
We shall refer to \eqref{eq:shift} as the shift relation. Recall
that if $(l,r)$ is Serre and moreover $(l,r)$ are inverse to each other,
then we say that $(l,r)$ is a Serre duality.

We say that $r : Q \rto Q$ is a Serre map if $(r,r)$ is a Serre Galois
connection.
\begin{proposition}
  \label{prop:SerreGC1}
  If $(l,r)$ is a Serre Galois connection, then
  $j \eqdef l \circ r = r \circ l$ is a nucleus and the restriction of
  $(l,r)$ to $Q_{j}$ yields a Frobenius quantale structure on
  $(Q_{j},\qmult_{j})$.
\end{proposition}
\begin{proof}
  We firstly prove that $x \qmult j(y) \leq j(x \qmult y)$, that is,
  $x \qmult rl(y) \leq rl(x \qmult y)$. Using the shift relation, this
  inclusion is equivalent to $l(x \qmult y) \qmult x \leq
  lrl(y)$. Since $lrl(y) = l(y)$, the latter inclusion is
  $l(x \qmult y) \qmult x \leq l(y)$ and equivalent, by the shift
  relation, to $x \qmult y \leq rl(x \qmult y)$, where the last
  inclusion holds since $r\circ l$ is a closure operator.
  The inclusion $j(x) \qmult y \leq j(x \qmult y)$ is proved similarly,
  using now $j = l \circ r$.

  Next, the restrictions of $(l,r)$ to $Q_{j}$ have $Q_{j}$ as codomain:
  if $x \in Q_{j}$, then $j(r(x)) = r(l(r(x))) = r(x)$, so
  $r(x) \in Q_{j}$. Similarly, $l(x) \in Q_{j}$.  These restrictions
  are also inverse to each other: for $x \in Q_{j}$,
  $l(r(x)) = r(l(x)) = j(x) = x$.  Let us argue that these
  restrictions are Serre: for $x,y ,z \in Q_{j}$,
  $x \qmult_{j} y = j(x \qmult y) \leq r(z)$ iff $x \qmult y \leq r(z)$
  (since $r(z) \in Q_{j}$) iff $z \qmult x \leq l(y)$ iff
  $z \qmult_{j} x = j(z \qmult x)\leq l(y)$, so the shift relation holds
  in $Q_{j}$.
\end{proof}

\begin{proposition}
  \label{prop:SerreGC2}
  Let $j$ be a nucleus on $(Q,\qmult)$ and suppose that
  $l,r : Q_{j} \rto Q_{j}$ is a Serre duality on
  $(Q_{j},\qmult_{j})$. Then $(l \circ j,r \circ j)$ is a Serre Galois
  connection on $(Q,\qmult)$ and the Frobenius quantale structure
  $(Q_{j},\qmult_{j},l,r)$ is induced from $(l \circ j,r \circ j)$ as
  described in Proposition~\ref{prop:SerreGC1}.
\end{proposition}
\begin{proof}
  Let $j,l,r$ be as stated. For $x,y,z \in Q$, we have
  $x \leq r(j(y))$ iff $j(x) \leq r(j(y))$ iff $j(y) \leq l(j(x))$ iff
  $y \leq l(j(x))$.
  Moreover
  \begin{align*}
    x \qmult y \leq r(j(z)) & \Tiff j(x \qmult y) = j(x) \qmult_{j} j(y)
    \leq
    r(j(z)) \\
    & \Tiff j(z \qmult x) = j(z) \qmult_{j} j(x) \leq l(j(y)) \Tiff z \qmult
    x \leq l(j(y))\,.
  \end{align*}
  Notice now that $j \circ r = r$ and $j \circ l = l$, since by
  assumption the codomain of $r$ and $l$ is $Q_{j}$. Then
  \begin{align*}
    (l \circ j) \circ (r \circ j) & = l \circ r \circ j = j
  \end{align*}
  and, similarly, $(r \circ j) \circ (l \circ j) = j$.  Finally, the
  restriction of $r \circ j$ (resp., $l \circ j$) to $Q_{j}$ is $r$
  (resp., $l$), for example, for $x \in Q_{j}$, we have
  $(r \circ j)(x) = r(j(x)) = r(x)$.
\end{proof}

We add next some  remarks.
\begin{definition}
  If $0$ is an element of a quantale $(Q,\qmult)$, then we say that $0$
  is \emph{weakly cyclic} if
  \begin{align*}
    0 \rlimpl (x \lrimpl 0) & = (0 \rlimpl x) \lrimpl 0\,, \quad \text{for
      all $x \in Q$.}
  \end{align*}
  We say that a Serre Galois connection $(l,r)$ on $(Q,\qmult)$ is
  representable by $0$ if
  \begin{align*}
    r(x) & = x \lrimpl 0\,,\qquad l(x) = 0 \rlimpl x\,, \qquad \text{for all
      $x \in Q$.}
  \end{align*}
\end{definition}
Observe that if a Serre Galois connection $(l,r)$ is representable by
$0\in Q$, then $0$ is a weakly cyclic element.

\begin{lemma}
  If $(Q,\qmult)$ is unital, then every Serre Galois connection is  
  representable.
\end{lemma}
\begin{proof}
  If $1$ is the unit of $(Q,\qmult)$, then
  \begin{align*}
    r(y) & = r(y) \rlimpl 1 = y \lrimpl l(1)
    \qquad \text{and}\qquad
    l(y)  = 1 \lrimpl l(y) = r(1) \rlimpl  y\,.
  \end{align*}
  Moreover
  \begin{align*}
    r(1) & = r(1) \rlimpl 1 = 1 \lrimpl l(1) = l(1)\,,
  \end{align*}
  and therefore $(l,r)$ is representable by $0 = r(1) = l(1) \in
  Q$. 
\end{proof}
\begin{lemma}
  \label{lemma:quotients}
  Let $0$ be a weakly cyclic element of a quantale $(Q, \qmult)$ and
  set $r(x) \eqdef x \lrimpl 0$ and $l(x) \eqdef 0 \rlimpl x$. Then
  $(l,r)$ is a representable Serre Galois connection. If $0$ is
  $j$-closed (with $j = r \circ l = l \circ r$), then
  $(Q_{j},\qmult_{j})$ is unital. If $(Q, \qmult)$ is unital, then $0$
  is $j$-closed.
\end{lemma}
\begin{proof}
  The first statement is obvious.  If $0$ is $j$-closed, then the
  relations $r(x) = x \lrimpl 0$ and $l(x) = 0 \rlimpl x$ hold within
  $Q_{j}$, whence $0$ is a dualizing element of $Q_{j}$ which
  therefore is unital.  Finally, if $1$ is the unit of $(Q,\qmult)$,
  then $0 = r(1) = rlr(1)= l(1) = lrl(1)$, showing that $0$ is
  $j$-closed.
\end{proof}

\begin{lemma}
  \label{lemma:commuting}
  For $(l,r)$ a Galois connection on $Q$, the condition
  $l \circ r = r \circ l$ holds if and only if the images of $l$ and
  $r$ in $Q$ coincide.
\end{lemma}
\begin{proof}
  For a map $f$ let $Img(f)$ denote its image.  Notice that
  $Img(r) = Q_{r\circ l}$, since $r(x) = r(l(r(x)))$ and similarly
  $Img(l) = Q_{l\circ r}$. Therefore, if $r \circ l = l \circ r $,
  then $Img(r) = Img(l)$. Conversely, if $Img(r) = Img(l)$, then
  $Q_{r\circ l} = Q_{l\circ r}$ and since a closure operator is
  determined by the set of its fixed points, then
  $l \circ r = r \circ l$.
\end{proof}

\begin{Example}[Trivial quantales]
  As we do not require units, each
  complete lattice $Q$ can be endowed with the trivial quantale
  structure, defined by:
  \begin{align*}
    x \qmult y & \eqdef \bot\,, \quad \text{for each $x,y \in Q$.}
  \end{align*}
  The trivial quantale structure clearly is not unital, unless $Q$ is
  a singleton.  Notice that, for the trivial structure,
  \begin{align*}
    x \lrimpl y = y \rlimpl x & = \top\,, \quad \text{for each
      $x,y \in Q$.}
  \end{align*}
  If $Q$ is autodual, say with $l,r : Q \rto Q$ inverse to each other,
  then $(Q,\qmult,l,r)$ is a unitless Frobenius quantale.
\end{Example}

The trivial structure is a source of counter-examples to quick
conjectures.  Consider $(\two,\qmult)$, the two-element Boolean algebra
equipped with the trivial quantale structure.  Then we have
$x \qmult y \leq r(z)$ if and only if $y \qmult z \leq l(x)$ no matter how
we choose $(l,r)$. Thus, we might chose $(l,r)$ not antitone and yet
satisfying the shift relation (take $r = l$ be the identity of
$\two$), or choose $(l,r)$ antitone but not a Galois connection (take
$r = l$ be the constant function taking $\bot$ as its unique value).

Also, consider the diamond lattice with 3 atoms (see
figure~\ref{fig:pentagonDiamond}) and endow it with the trivial
quantale structure.
Take for $r$ a duality that cyclically permutes the 3 atoms and let
$l$ be its inverse, so to have $l \neq r$. Then $(l,r)$ is a Serre
duality on the commutative quantale over $M_3$. On the other hand, every
representable Serre duality on a commutative quantale is necessarily
self-adjoint (i.e. we have $r = l$).

It might be asked whether $(l,r)$ representable on $Q$ implies that
$(l,r)$ is representable on $Q_{j}$. This is not the case. On the
chain $0 < 1 < 2$, consider the following commutative quantale structure:
$$
\begin{array}{r|ccc}
  \qmult & 0 & 1 & 2 \\
  \hline
  0 & 0 & 0 & 0 \\
  1 & 0 & 0 & 0 \\
  2 & 0 & 0 & 1 
\end{array}
$$
By commutativity, $0$ is weakly cyclic and, letting
$r(x) = x \lrimpl 0$, we have $r(0) = r(1) = 2$ and $r(2) =
1$. Therefore $Q_{j} = \set{1,2}$ and the $\qmult_{j}$ is the trivial
quantale structure on $Q_{j}$. As the implication of the trivial
structure is constant, it cannot yield the (unique) duality of
$\set{1,2}$. Considering Lemma~\ref{lemma:quotients}, notice here that
$0$ is not $j$-closed.

It is not difficult to reproduce similar counter-examples with the
multiplication being non-trivial (i.e. $x \qmult y \neq \bot$, for some
$x,y \in Q$).

\Myparagraph{Phase quantales.}
We explore next the construction described in
Propositons~\ref{prop:SerreGC1} and~\ref{prop:SerreGC2} when
$(Q,\qmult)$ is the free quantale $(P(S), \bullet)$ over a semigroup
$(S,\cdot)$. Recall (see e.g. \cite{EGHK2018}) that
\begin{align*}
  X\bullet Y & \eqdef \set{ x \cdot y  \mid  x\in X, \,y\in Y}\,,
  \\
  X\lrimpl Y & = \set{s \in S  \mid  s \cdot x \in Y, \text{ for all } x\in
    X}\,, \\
  Y\rlimpl X & = \set{ s\in S  \mid  x\cdot s \in Y, \text{ for all } x\in X}\,.
\end{align*}
It is well-known (see e.g. \cite{Ore1944}) that Galois connections
$(l,r)$ on $P(S)$ bijectively correspond to binary relations $R$ via the
correspondence
\begin{align*}
  x R y & \Tiff x \in l(\,\set{y}\,) \Tiff y \in r(\,\set{x}\,)
  \intertext{so} 
  r(Z) & = \set{ u \in S \mid z R u, \text{ for all } z \in Z }\,,
  \quad l(Y) = \set{ u \in S \mid u R y, \text{ for all } y \in Y
  }\,.
\end{align*}

We aim ro characterize Serre Galois connections via properties of
their corresponding relation.
\begin{proposition}
  A Galois connection $(l,r)$ on $(P(S), \bullet)$
  satisfies the shift relation~\eqref{eq:shift} if and only if its corresponding binary
  relation $R$ satisfies 
  \begin{align}
    \label{eq:defRassoc}
    x \cdot y R z   \Tiff x R y \cdot z  \,, \quad
    \text{for
    all $x,y,z \in S$}.
  \end{align}

\end{proposition}
\begin{proof}
  If $(l,r)$ satisfie ~\eqref{eq:shift}, then the following is a
  chain of equivalent statements 
  (where $X,Y,Z$ are subsets of $S$):
  \begin{align*}
    \forall x \in X,y \in Y , z \in Z, \;z R x \cdot y & \Tiff
    X \bullet Y  \subseteq r(Z)\\
    & \Tiff
    Z \bullet X  \subseteq l(Y)\\
    & \Tiff \forall x \in X,y \in Y , z \in Z,\; z \cdot x R y\,.
  \end{align*}
  Considering singletons we obtain that $R$ satisfies
  \begin{align*}
    \forall x,y,z \in S\;\; z R x \cdot y \Tiff z \cdot x R y\,,
  \end{align*}
  and clearly the latter condition suffices for \eqref{eq:shift} to
  hold.
\end{proof}
We call \emph{associative} a relation satisfying~\eqref{eq:defRassoc}.
To complete the characterisation of Serre Galois connections, we need
to characterise the condition $l \circ r = r\circ l$ in terms of the
corresponding relation $R$. Notice first that the condition trivially
holds if $R$ is symmetric, in which case $r = l$.

By Lemma~\ref{lemma:commuting}, $r\circ l = l \circ r$ if and only if
$Im(r) =Im(l)$, that is
\begin{align*}
  \forall X \exists Y\, \Tst r(X) = l(Y)\,,\quad\tand\quad
  \forall Y \exists X\, \Tst r(X) = l(Y)\,.
\end{align*}
Clearly, these two conditions hold if and only if they hold with $X$
(in the fist case, or $Y$ in the second) restricted to singleton,
yielding
\begin{align}
  \label{cond:comm1}
  &\forall x \exists Y_{x} \forall z \; (x R z \Tiff z R y,\, \forall
  y \in Y_{x})\,, \\
  \label{cond:comm2}
  & \forall y \exists X_{y} \forall z \; (z R y \Tiff x R z,\, \forall
  x \in X_{y})\,.
\end{align}
We call \emph{weakly-symmetric} a relation
satisfying~\eqref{cond:comm1} and~\eqref{cond:comm2}.

\medskip

We exemplify our previous observations, in particular the use of
Proposition~\ref{prop:SerreGC1}, in different ways.
\begin{Example}
  Our first example illustrates how to construct ad-hoc unitless
  Frobenius quantales with almost no effort.
  
  Consider $(\Sigma^{+},\cdot)$, the free semigroup over the alphabet
  $\Sigma$.
  Let $R$ be the binary relation total on letters and empty otherwise:
  $w R u$ iff $w,u \in \Sigma$. Due to the absence of units, condition
  \eqref{eq:defRassoc} is trivially satisfied. Notice also that $R$ is
  symmetric.
  The closed subsets of $\Sigma^{+}$ are $\emptyset$, $\Sigma$, and
  $\Sigma^{+}$ and the table for the multiplication is as follows:
  $$
  \begin{array}{r|ccc}
    &\emptyset & \Sigma & \Sigma^{+} \\
    \hline
    \emptyset & \emptyset & \emptyset & \emptyset \\
    \Sigma & \emptyset & \Sigma^{+}& \Sigma^{+}\\
    \Sigma^{+} & \emptyset &\Sigma^{+}& \Sigma^{+}
  \end{array}
  $$
  This quantale does not have units but it has a Serre duality, the
  unique duality of the chain
  $\emptyset \subseteq \Sigma \subseteq \Sigma^{+}$.
\end{Example}

\begin{Example}[Unital Frobenius quantales from pregroups]
  Here we illustrate how the same tools can be used to build Frobenius
  quantales that are not Girard quantales.
  Recall that a pregroup, see e.g. \cite{Buszkowski2001}, is an
  ordered monoid $(M,1,\cdot, \leq)$ coming with functions
  $l,r : M \rto M$ satisfying
  \begin{align*}
    x \cdot x^{r} &\leq 1\,, & x^{l} \cdot x & \leq 1\,,&
    1& \leq x^{r} \cdot x\,,& 1 \leq  x
    \cdot x^{l}\,,
  \end{align*}
  for all $x \in M$.  That is, a pregroup is a posetal rigid
  category. Consider the relation $R$ defined by $x R y$ iff
  $x \cdot y \leq 1$. Clearly, $R$ is associative.  Observe now that
  $x \cdot z \leq 1$ if and only if $z \cdot x^{r}{}^{r}\leq 1$. That
  is, the condition in \eqref{cond:comm1} is satisfied by letting
  $Y_{x} = \set{x^{r}{}^{r}}$. Similarly, \eqref{cond:comm2} is
  satisfied by letting $X_{y} = \set{y^{l}{}^{l}}$. Since $M$ has a
  unit, both $P(M)$ and $P(M)_{j}$ are unital.
\end{Example}

\begin{Example}[Unitless Girard quantales from \Cstar-algebras]
  Let us consider an algebra $A$ coming with an associative symmetric
  pairing (i.e. a bilinear form) $\pairing{-,-}$ into the base field
  $\K$.  We do not assume that $A$ has a unit. We recall that a
  pairing is said to be associative if it satisfies
  $\langle x\cdot y,z\rangle = \langle x,y\cdot z\rangle$, for each
  $x,y,z \in A$ (in which $A$ is called a Frobenius algebra).

  The binary relation $R$, defined by $xR y$ if and only if
  $\langle x,y\rangle = 0$, is then an associative relation on the
  semigroup reduct $(A,\cdot)$ of $A$, so we can consider the powerset
  quantale $(P(A),\bullet)$ and the Serre Galois connection $(l,r)$ the
  relation $R$ gives rise to. Since $R$ is symmetric, we denote
  $l(X) = r(X) = \Rneg{X}$, as usual from standard algebra, and so
  $j(X) = \Dneg{X}$. Clearly, if $A$ has a unit, then $P(A)$ and its
  quotient $P(A)_{j}$ have a unit as well.

  \medskip
  
  We argue next that the converse holds, if we can transform the
  pairing into a sort of inner product (as for example with
  \Cstar-algebras). Namely, suppose that $A$ comes with an involution
  $\Star{\fun} : A \rto A$ such that, for each $f \in A$,
  \begin{align}
    \label{cond:preinner}
    \pairing{f,\Star{f}} = 0 & \Timplies f = 0\,.
  \end{align}
  In particular, assuming \eqref{cond:preinner}, we have
  $\Rneg{A} = \set{0}$. Under these assumptions, the following
  statement holds:
  \begin{proposition}
    If $P(A)_{j}$ has a unit, then $A$ has a
    unit.
  \end{proposition}
  \begin{proof}
    Obviously, every set of the form $\Rneg{X}$ is a subspace of
    $A$. Let us say that a subpace is closed if it is of the form
    $j(Y)$ or, equivalently, $\Rneg{X}$ (for some $X$ or some $Y$).

    We firstly observe that every one dimensional subspace of $A$ is
    closed. That is, we have $j(\set{f}) = \K f$, for each
    $f \in A \setminus \set{0}$. Indeed, for $g \in A$ and
    $k = \frac{\pairing{g,\Star{f}}}{\pairing{f,\Star{f}}}$, we can
    write $g = (g - kf) + kf$ with $g - kf \in
    \Rneg{\set{\Star{f}}}$. Because of \eqref{cond:preinner}, we also
    have $f \not \in \Rneg{\set{\Star{f}}}$. That is, we have and
    $\K f \vee \Rneg{\set{\Star{f}}} = A$ and
    $\K f \land \Rneg{\set{\Star{f}}} = \set{0}$ in the lattice of
    subspaces of $A$.  The relation
    $\K f \vee \Rneg{\set{\Star{f}}} = A$ a fortiori implies the
    relation $j(\K f) \vee_{j} \Rneg{\set{\Star{f}}} = A$, taken in
    the lattice of closed subpaces of $A$.  Using the duality, we
    derive
    $\set{0} = \Rneg{A} = \Rneg{(j(\K f) \vee_{j}
      \Rneg{\set{\Star{f}}})} = \Rneg{j(\K f)} \land_{j}
    \Dneg{\set{\Star{f}}} = \Rneg{\set{f}} \land
    \Dneg{\set{\Star{f}}}$, where we have used that both meets
    $\land$ (the meet in the lattice of subspaces of $A$) and
    $\land_{j}$ (the meet in the lattice of closed subspace of $A$)
    are computed as intersection.
    Then we also have $\set{0} = \Rneg{\set{\Star{f}}} \land
    \Dneg{\set{\SStar{f}}} = \Rneg{\set{\Star{f}}} \land
    \Dneg{\set{f}}$.
    It follows that both $\K f$ and $j(\set{f}) = \Dneg{\set{f}}$ are
    complements of $\Rneg{\set{\Star{f}}}$ in the lattice of
    subspaces of $A$.  Since $\K f \subseteq j(\set{f})$, we derive
    $\K f = j(\set{f})$ using modularity of this lattice.

    Let now $U$ be the unit of $P(A)_{j}$. Then
    $U \bullet \set{f} = j(U) \bullet \set{f} \subseteq j(U \bullet
    \set{f}) \subseteq j(U \bullet j(\set{f})) = U \bullet_{j}
    j(\set{f}) = j(\set{f}) = \K f$ and, similarly,
    $\set{f} \bullet U \subseteq \K f$.
    Thus, for each $u \in U$ and $f \in A$, $u \cdot f = k f$ and
    $f \cdot u = k' f$ for some $k,k' \in \K$. In particular, for
    $u,v \in U$, then $k u = u \bullet v = k' v$, showing that $U$ has
    dimension $1$ (the unit $U$ cannot be $\set{0}$ unless
    $A = \set{0}$).  We claim now that, given $u \in U$, there exists
    $k \in \K$ such that, for all $f \in A$, $u \cdot f = k f$.
    Indeed, let $f,g \in A \setminus \set{0}$ and write
    $u\cdot f = k_{f}f $ and $u \cdot g = k_{g} g$.  If $f = k g$,
    then
    $k_{f}f = u \cdot f = u \cdot k g = k (u \cdot g) = kk_{g}g =
    k_{g}kg = k_{g} f$, thus $k_{f} = k_{g}$.  Suppose now that $f,g$
    are linearly independent.  Then, for some $k \in \K$,
    $k f + k g = k (f + g) = u\cdot (f + g) = u \cdot f + u \cdot g =
    k_{f}f + k_{g}g$ which yields the relation
    $(k- k_{f})f = (k_{g} - k)g$ and $k_{f} = k = k_{g}$.
    We have proved the claim from which it readily follows that
    $u' \eqdef \frac{1}{k} u$ is a left unit of $A$. Similarly, we can
    find a right unit $\tilde{u}$ and, as usual, $\tilde{u} = u'$.
    Thus $A$ is unital.
  \end{proof}

  The \Cstar-algebra of $n\times n$ matrices over the complex numbers
  $\C$ comes with the pairing $\pairing{A,B} = tr(A \cdot B)$, where
  $\Star{A}$ is the conjugate transpose of $A$. This algebra is unital
  and gives rise to a well known Girard quantale, see e.g \cite[\S
  2.6.15]{EGHK2018}.
  We can adapt this construction to consider classes of linear
  operators on an infinite dimensional Hilbert space $H$. A continuous
  linear mapping $f : H \rto H$ is trace class, see e.g. \cite[\S
  18]{Conway1990}, 
  if
  \begin{align*}
    \sum_{e \in {\mathcal{E}}} \innerH{|f|e,e} < \infty\,,
  \end{align*}
  where $\innerH{-,-}$ above is the inner product of the Hilbert
  space $H$, ${\mathcal{E}}$ is an orthonormal basis, and $|f|$ is the
  unique positive operator such that $\Star{f} \circ f = |f|^{2}$. For
  such an operator, we can define
  \begin{align}
    tr(f) & \eqdef \sum_{e \in {\mathcal{E}}} \innerH{f(e),e}\,,
    \qquad \pairing{f,g} \eqdef tr(f\cdot g)\,,
    \label{def:tracePairing}
  \end{align}
  yielding an associative symmetric pairing.  Now, $f$ is trace class
  if and only if its adjoint $\Star{f}$ is trace class. With respect
  to the involution given by the adjoint,
  this pairing
  satisfies
  \eqref{cond:preinner}.
  The algebra of trace class operators does not have a unit. First of
  all, it is easily seen that the identity is not trace class. Indeed,
  this algebra is a well-known proper (when $H$ is infinite
  dimensional) ideal of the algebra of bounded linear operators on
  $H$. Most importantly, this algebra cannot have a unit. In fact, for
  each $h \in H$ there exists a trace class operator $c_{h}$ and
  $p_{h} \in H$ such that $c_{h}(p_{h}) = h$. 
  For example, if $\mathcal{E}$ is a basis for $H$, we can let
  $e_{0} \in \mathcal{E}$, and define $c_{h}(e_{0}) = h$, and
  $c_{h}(e) = 0$, for $e \in \mathcal{E} \setminus \set{e_{0}}$.  
  Then a unit $u$ is
  forcedly the identity:
  \begin{align*}
    u(h) & = u(c_{h}(p_{h})) = (u \circ c_{h})(p_{h}) = c_{h}(p_{h}) =
    h\,, \qquad\text{for all $h \in H$.}
  \end{align*}
  We collect these observations into a formal statement.
  \begin{theorem}
    The collection of closed subspaces of the algebra of trace class
    operators is a unitless Girard quantale.
  \end{theorem}
  Similar considerations can be developed for Hilbert-Schmidt
  operators, see \cite{Conway1990}. Indeed, the composition of two
  such operators is trace class, showing that the formula for the
  pairing in \eqref{def:tracePairing} also applies to these operators.
\end{Example}

Our next goal is to give a representation theorem, analogous to the
representation theorem for Girard quantales, see e.g. \cite[Theorem
2]{Rosenthal1990b}. The representation theorem is here extended to
Frobenius and unitless quantales.
\begin{proposition}
  Given a Frobenius quantale $(Q, \ast, \Lneg{\fun},\Rneg{\fun})$,
  define $x R y$ iff $x \leq \Lneg{y}$.  Then $R$ is an associative
  weakly symmetric relation yielding a Serre Galois connection $(l,r)$ on
  $(P(Q), \bullet)$, whose quotient $(P(Q)_{j},\bullet_{j},l,r)$ is
  isomorphic to $(Q, \ast, \Lneg{\fun},\Rneg{\fun})$.
\end{proposition}
\begin{proof}
  Let us verify that $R$ is associative:
  \begin{align*}
    z R x \ast y & \Tiff z \leq \Lneg{(x \ast y)}
    \Tiff x \ast y \leq \Rneg{z} 
    \Tiff z \ast x \leq \Lneg{y}
    \Tiff
    z \ast x R y \,.
  \end{align*}
  Let $(l,r)$ be the Galois connection corresponding to $R$.  Let us
  argue that $r \circ l = l \circ r$ by showing that the images of $l$
  and $r$ coincide, cf. Lemma~\ref{lemma:commuting}.  By definition, a
  subset $X \subseteq Q$ in the image of $l$ is of the form
  $ \set{y \mid y \leq \Lneg{x}, \text{ for all } x \in X}$.
  As the condition $y \leq \Lneg{x}, \text{ for all } x \in X$, is
  equivalent to  $y  \leq \bigwedge_{x \in X} \Lneg{x} = \Lneg{(\,\bigvee X\,)}$,
  we deduce that such set equals $\downset \Lneg{(\,\bigvee X\,)}$.
  Since every element $y \in Q$ is of the form $\Lneg{z}$ for some
  $z \in Q$, we deduce that the image of $l$ is the set of principal
  downsets of $Q$. Similalry, a
  subset $X \subseteq Q$ in the image of $r$ is of the form
  $ \set{y \mid y \leq \Rneg{x}, \text{ for all } x \in X} =  \downset
  \Rneg{(\,\bigvee X\,)}$ and the image of $r$ is again the set of
  principal downsets of $Q$.

  Finally, $j(X) = \downset \bigvee X$, since
  \begin{align*}
    j(X) = l(r(X)) & =\; \downset (\Lneg{\bigvee (\downset \Rneg{(\,\bigvee
        X\,)})})
    = \downset( \Lneg{\Rneg{(\,\bigvee X\,)}}) = \downset (\,\bigvee X\,)\,.
  \end{align*}
  It is then easily seen that the mapping $\downset : Q \rto P(Q)_{j}$
  is inverted by $\bigvee\, : P(Q)_j \rto Q$. These maps are
  quantale homomorphisms, indeed we have, for all $x$ and $y$ in $Q$,
  \begin{align*}
    \downset{x} \,\,\bullet_j \,\downset{y} & =
    \,\downset{(\bigvee(\downset{x}\,\bullet \downset{y})) = \{z\in Q \mid
      z\leq x\ast y\}} =\, \,\downset{(x\ast y)}.
    \tag*{\qedhere}
  \end{align*}
\end{proof}

\color{black}

\section{The Girard quantale of tight maps}

We present in this section the main example of a unitless Girard
quantale, the one that prompted this research.

We denote by $L^L$ the lattice of all function from $L$ to $L$,
by $\LLs$ the lattice of all join-preserving endomaps of $L$, and by
$\LLi$ the lattice of all meet-preserving endomaps of $L$.
For a function $f: L\rto L$, its Raney's transforms $\rans{f}$ and
$\rani{f}$ are defined by
\begin{align*}
  \rans{f}(x) & \eqdef \bigvee_{x\nleq t}f(t)\,, 
  &
  \rani{f}(x) & \eqdef \bigwedge_{t\nleq x}f(t)\,.
\end{align*}
Let us remark that, in the definition above, we do not require that
$f$ has any property, such as being monotone. 

\begin{definition}
  An endomap $f : L \rto L$ is \emph{tight} if $f=\rand{f}$. We write
  $\LLst$ for the set of tight endomaps of $L$.  We say that
  $f : L \rto L$ is \emph{\cotight} if $f = \ranD{f}$ and write
  $\LLit$ for the set of \cotight maps from $L$ to $L$.
\end{definition}

Most of the properties of the Raney's transforms appear already in
\cite{HiggsRowe1989}, we add proofs of lemmas. We aim at
demonstrating the following theorem.
\begin{theorem}\label{tight_girard}
  For any complete lattice $L$, the tuple $(\LLt, \circ, \Star{\fun})$
  is a Girard quantale (as defined in Definition~\ref{defi_B}), 
  where $\circ$ is the function composition and the duality $\Star{\fun}$
  defined by 
  $\Star{f} \eqdef \rans{\radj{f}}$.
\end{theorem}

In order to prove the theorem, let us recall some elementary
properties of Raney's transforms. In the following, $L$ shall be an
arbitrary but fixed complete lattice.
\begin{lemma}
  \label{lemma:jp}
  For any function $f : L \rto L$,  $\rans{f}$ has a \ra, 
  and, in particular $\rans{f}$ is \jp and monotone. If $f$ is \mp,
  then the following relation holds:
  \begin{equation}\label{right_adjoint_rans}
    \rho(\rans{f})=\rani{(\ladj{f})}\,.
  \end{equation}
  For $f : L \rto L$, $\rani{f}$ is \mp and, when $f$ is \jp, his \la
  satisfies
  \begin{equation}\label{left_adjoint_rani} 
    \ladj{\rani{f}} = \rans{(\radj{f})}\,.
\end{equation}
\end{lemma}
\begin{proof}
  For each $x, y \in L$, we have
  \begin{align*}
    \rans{f}(x) = \bigvee_{x \not\leq t} f(t) \leq y \Tiff & \text{for
      all $t$}\,,\; x \not\leq t \Timplies f(t) \leq y\,, \\
    \Tiff & \text{for all $t$}\,,\; f(t) \not\leq
    y \Timplies x \leq t \,, \\
    \Tiff &  x \leq \bigwedge_{f(t) \not\leq y } t\,.
  \end{align*}
  Therefore, if we define
  $g(y) \eqdef \bigwedge_{f(t) \not\leq y } t$, then $g$ is \ra to
  $\rans{f}$. Let us argue that equality \eqref{right_adjoint_rans}
  holds if $f$ is \mp. 
  We shall show that, for all $x$ and $y$
  in $L$, we have
  $$
  \rans{f}(x)\leq y \Tiff 
  x\leq \rani{(\ladj{f})}(y)\,.
  $$
  On the one hand we have
\begin{align}
  \nonumber
  \tag{a}
  \bigvee_{x\nleq t}f(t) \leq y \Tiff \text{ for all } t\in L \st
  x\nleq t, \,f(t) \leq y\,.
\end{align}
On the other hand we have
\begin{align}
  \nonumber
  \tag{b}
  x\leq \bigwedge_{t\nleq y}\ladj{f}(y) \Tiff \text{ for all } t\in L \st
  t\nleq y, \,x \leq \ladj{f}(t).
\end{align}
For the implication \Impl{a}{b}, let $t \in L$ be such that
$t\nleq y$.  Suppose that $x\nleq \ladj{f}(t)$. Then, by (a),
$f(\ladj{f}(t)) \leq y$ and, by the adjunction,
$t \leq f(\ladj{f}(t))$, whence we deduce $t \leq y$,
\contr. Therefore $x\leq \ladj{f}(t)$. The implication \Impl{b}{a} is
proved in a similar way.

The second part of the statement follows by duality.
\end{proof}

Recall that the set $\LL$ is pointwise ordered, i.e. we have
$f \leq g$ iff $f(t) \leq g(t)$, for all $t \in L$.
\begin{proposition}\label{prop:raniadjointrans}
  The operation $\rani{(-)}: \LL \rto \LL$ is \ra to
  $\rans{(-)}:\LL\rto \LL$.
\end{proposition}
\begin{proof}
  Let $f$ and $g$ be two endomaps of $L$, we have
\begin{align*}
  \rans{f}\leq g &\Tiff \text{ for all } x\in L, \;\rans{f}(x) =
  \bigvee_{x\nleq t}f(t) \leq g(x) \\
  &\Tiff \text{ for all } x,t\in L,\;
  x\nleq t \Timplies f(t)\leq g(x) \\
  &\Tiff \text{ for all } t\in L,\;
  f(t)\leq \bigwedge_{x\nleq t} g(x)  = \rani{g}(t)
  \Tiff f\leq \rani{g}\,.\tag*{\qedhere}
\end{align*}
\end{proof}

Let us briefly analyse the typing of the Raney's
transforms. Lemma~\ref{lemma:jp} exhibits the Raney's transforms as
overloaded operators: for example, the codomain of $\rans{\fun}$ can
be taken as any of the complete lattices among $\LL$, $\LLs$, $\LLst$,
as suggested in the left diagram in
Figure~\ref{fig:typingRaney}. 

To further analyse the situation, recall that a map $f : L \rto L$ is
\emph{\cotight} if $f = \ranD{f}$ and that we use use $\LLit$ for the
set of \cotight maps from $L$ to $L$. 
The adjunction $\rans{\fun} \adj \rani{\fun}$ of
Lemma~\ref{prop:raniadjointrans} restricts from $\LL$ as suggested in
the diagram on the left of Figure~\ref{fig:typingRaney}. This is
further illustrated in the diagram on the right of the figure where
the vertical inclusions on the left are \mp and the vertical
inclusions on the right are \jp. The two Raney's transforms in the
bottom row are inverse isomorphisms, as a standard consequence of the
characterisation of factorizarion systems in the category of
sup-lattices and \jp maps, see~\cite{EGHK2018}.

\smallskip

We recall next some usual consequences of the adjunction $\rans{\fun}
\adj \rani{\fun}$ established in Proposition~\ref{prop:raniadjointrans}.
\begin{lemma}
  \label{lemma:image}
  A map $f : L \rto L$ is tight if and only if it lies in the image of
  the Raney's transform $\rans{\fun}$.
\end{lemma}
\begin{proof}
  By general properties of adjunctions, $\rans{f} = \rans{f}\rand{}$,
  showing that if $f$ lies in the image of $\rans{\fun}$, then it is
  tight. The converse is obvious.
\end{proof}
\begin{lemma}
  \label{lemma:greatestAndLeast}
  For each map $f : L \rto L$, $\rand{f}$ is the greatest tight map
  below $f$, and $\ranD{f}$ is the least \cotight map above $f$.
\end{lemma}
\begin{proof}
  We have $\rand{f} \leq f$, and if $g$ is tight and below $f$, then
  $g = \rand{g} \leq \rand{f}$.
\end{proof}

\begin{figure}
  \centering
  \includegraphics{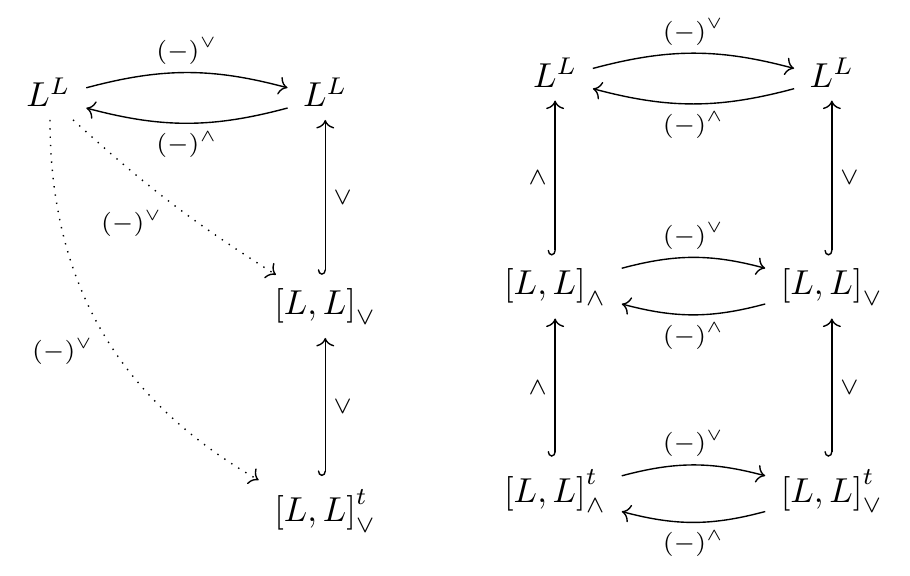}
  \caption{Typing of the Raney's transforms}
  \label{fig:typingRaney}
\end{figure}
\begin{corollary}
  \label{cor:closed_suprema}
  For every complete lattice $L$, the set $\LLt$ is closed under
  arbitrary suprema.
\end{corollary}
\begin{proof}
  Proposition~\ref{prop:raniadjointrans} implies that the operation
  $\rand{\fun}$ is an interior operator, that is, the dual of a
  closure operator. As the set of fixed-points of a closure operator
  is closed under arbitrary infima, the set of fixed-points of an
  interior operator is closed under arbitrary suprema.
\end{proof}

\begin{lemma}
  If $f$ is \jp and  $g$ is any function, then
  \begin{align}
    \rans{(f \circ g)} & = f \circ \rans{(g)} 
    \,.
    \label{eq:rani_sort_parenthese}
  \end{align}
\end{lemma}
\begin{proof}
  Let $f$ and $g$ be as stated, and compute as follows:
  \begin{align*}  
    (f\circ \rans{g})(x) &  = f(\rans{g}(x))
    = f(\bigvee_{x\nleq t} g(t))
    = \bigvee_{x\nleq
      t}(f\circ g)(t) = \rans{(f\circ g)}(x)\,.
    \tag*{\qedhere}
  \end{align*}
\end{proof}

\begin{corollary}
  \label{cor:closed_composition}
  Tight maps are closed under composition.
\end{corollary}
\begin{proof}
  If $f$ and $g$ are tight maps, then $f$ is \jp, and so
  equation~\eqref{eq:rani_sort_parenthese} ensures that
  $f \circ g = f \circ \rand{g} = \rans{(f \circ \rani{g})}$, so
  $f \circ g$ belongs to the image of $\rans{\fun}$. As we have seen
  in Lemma~\ref{lemma:image}, this is enough to ensure that this map
  is tight.
\end{proof}

Corollaries~\ref{cor:closed_suprema} and~\ref{cor:closed_composition}
ensure that $(\LLt, \circ)$ is a subquantale of
$(\homm{L}{L}, \circ)$. The following proposition suffices to prove
that $(\LLt, \circ,\Star{\fun})$ is a Girard quantale, as stated in
Theorem~\ref{tight_girard}, where we recall that
$\Star{f} \eqdef \rans{\radj{f}}$.

\begin{proposition}
  The map $\Star{\fun}$ is an involutive Serre duality on $(\LLt, \circ)$.
\end{proposition}
\begin{proof}
  First the map $\Star{(-)}$ is an involution. Indeed using the
  equality~\eqref{right_adjoint_rans}, we have
  $$
  \SStar{f} = \rans{\rho(\rans{\radj{f}})} =
  \rans{\rani{\ladj{\radj{f}}}} = \rans{\rani{f}} =  f\,. 
  $$
  Since this map is the composition of a monotone map with an
  antitone one, it is antitone.  Let us now show that the shift
  relation~\eqref{eq:shift} holds, so $\Star{(-)}$ is a Serre duality.
  \begin{align*}
    f \circ g \leq \Star{h} 
    \Tiff & g \leq  \radj{f} \circ \Star{h} \\
    \Tiff & g \leq \rand{(\radj{f} \circ \Star{h})}\,, \tag*{since
      $\rand{k}$ is the greatest tight map below
      $k$, see Lemma~\ref{lemma:greatestAndLeast},}
    \\
    & = \rans{(\radj{f} \circ (\rani{\Star{h}}))} = \rans{(\radj{f}
      \circ \rho(h))} = \rans{\rho(h \circ f)} = \Star{(h \circ f)}\,,
    \tag{*}
    \\
    \Tiff & h \circ f \leq \Star{g}\,, \tag*{since
      $\Star{\fun}$ is an antitone involution,}
  \end{align*}
  where in (*) we have used the dual of equation
  \eqref{eq:rani_sort_parenthese}, the fact that $\ranD{g} = g$ if $g$
  is \cotight, the fact that $\rho(h)$ is \cotight when $h$ is tight
  (since $\rho(h) = \rho(\rand{h}) = \rani{\ladj{\rani{h}}}$), and usual
  properties of adjoints.
\end{proof}

Let us recall at this point \RT:
\begin{theorem}[Raney \cite{Raney60}]
  A complete lattice is \cd if and only if $id_{L}$ is tight.
\end{theorem}

The following statement appears in \cite{KrumlPaseka2008,EggerKruml2010} and, with few
subtle differences, in \cite{S_RAMICS2020}:
\begin{theorem}
  \label{thm:KP}
  The quantale $(\LLs,\circ)$ is a Girard quantale if and only if $L$
  is a \cdlatt.
\end{theorem}

We refine this result by arguing that the Girard quantale
$(\LLst,\circ,\Star{\fun})$ has a unit if and only if $L$ is a
\cdlatt, see Theorem~\ref{thm:TightIsUnital}. If the quantale of tight
maps is unital, then its unit necessarily is the identity, and then
every map $f \in \LLs$ is tight,
$f = f \circ id_{L} = f \circ \rand{(id_{L})} = \rans{(f \circ
  \rani{(id_{L})})}$. We obtain therefore Theorem~\ref{thm:KP} as a
corollary.

\smallskip

To achieve our goals, we recall some elementary
properties of the quantale $(\LLs,\circ)$.  In particular, a
characterisation of tights maps shall be given using the fundamental
morphisms $c_y$ and $a_x$.  For given $x,y \in L$, these maps are
defined by
\begin{align*}
  c_y(t) \eqdef \begin{cases}
    y\,, &t\neq \bot\,, \\
    \bot\,, &t = \bot\,,
  \end{cases}
  \quad\text{~~~ and ~~~}\quad
  a_x(t) \eqdef
  \begin{cases} \top\,, & t\nleq x\,, \\
    \bot\,, & t\leq x\,.
  \end{cases}
\end{align*}
These two maps are tight. Indeed, $c_{y} = \rans{(K_{y})}$, where
$K_{y} : L \rto L$ is the constant function with value $y \in L$, and
$a_{x} = \rans{(\delta_{x})}$ where $\delta_{x}(t) = \top$ if $t = x$
and, otherwise, $\delta_{x}(t) = \bot$.  Therefore $c_{y}$ and $a_{x}$
are tight and then, using Corollary~\ref{cor:closed_composition},
$c_y\circ a_x$ is tight as well. Let us remind the following fact:
\begin{lemma}
  \label{lemma:generated}
  The set of tight maps is generated under (arbitrary) joins by the maps $c_y\circ a_x$, for $x$ and $y$ in $L$.
\end{lemma}
\begin{proof}
  Recall that $f$ is a tight map if and only if there exists
  $g \in \LL$ such that $\rans{g} = f$. 
  Now for such a $g$, we have
  \begin{equation}\label{ransgalwaystighteq}
    \rans{g}(x) = \bigvee_{x\nleq t}g(t) = \bigvee\set{(c_{g(t)}\circ a_t)(x) \mid t\in L}\,,
  \end{equation}
  since
  \begin{align*}
    (c_{g(t)}\circ a_t)(x) & =
    \begin{cases}
      c_{g(t)}(\top) = g(t)\,, & x\nleq t\,,\\
      \bot\,,& x \leq t\,.
    \end{cases}
    \tag*{\qedhere}
  \end{align*}
\end{proof}

\begin{proposition}
  \label{prop:TightIsUnital}
  If $u$ is a left or right unit of the quantale
  $(\LLt, \circ)$, then $u$ is the identity of $L$. 
\end{proposition}
\begin{proof}
  Let $u \in \LLt$ be a left unit, then, for each $y \in L$,
  \begin{align*}
    u(y) & = u(c_{y}(\top)) = (u \circ c_{y})(\top) = c_{y}(\top) = y\,,
  \end{align*}
  so $u$ is the identity of $L$. In a similar way, if $u \in \LLt$ is a right
  unit,  then, for each $x,t \in L$,
  \begin{align*}
    t \leq x & \Tiff \bot = a_{x}(t)  =  (a_{x} \circ u)(t)
    \Tiff u(t) \leq x\,.
  \end{align*}
  Thus $u(t) = t$, for each $t \in L$, and
  we conclude again that $u$ is the identity of $L$.
\end{proof}

\begin{theorem}
  \label{thm:TightIsUnital}
  The Girard quantale $(\LLt, \circ, \Star{\fun})$ is unital if and
  only if $L$ is a \cdlatt.
\end{theorem}
\begin{proof}
  If $L$ is a \cdlatt, then $\LLt = \LL$ and $\LL$ is a Girard
  quantale \cite{KrumlPaseka2008,EGHK2018,S_RAMICS2020}. Conversely,
  If $\LLt$ has a unit, then by the previous Proposition this unit is
  the identity. Thus $\id_{L} \in \LLt$ and, by \RT, $L$ is a \cdlatt.
\end{proof}

\section{Tight maps from Serre Galois connections}

In this section we further illustrate the use of Serre Galois
connections by arguing that the Girard quantale of tight maps arises
from a Serre self-adjoint Galois connection on the set of \mp
functions $\LLi$.

If $f$ is a monotone map, then we let $\closure{f}$ be the least
$g \in \LLi$ such that $f \leq g$.

\begin{lemma}
  If $g$ is \jp, then
  \begin{align}
    \label{eq:cotightOne}
    \closure{g \circ \closure{f}} & = \closure{g \circ f}\,.
  \end{align}
\end{lemma}
\begin{proof}
  Since $f \leq \closure{f}$, then
  $\closure{g \circ f} \leq \closure{g \circ \closure{f}}$ by
  monotonicity.

  To derive the converse inclusion, it will be enough to
  show that $g \circ \closure{f} \leq \closure{g \circ f}$.  From
  $g \circ f \leq \closure{g \circ f}$ it follows
  $f \leq \rho(g) \circ \closure{g \circ f}$. Since
  $\rho(g) \circ \closure{g \circ f} \in \LLi$, we have
  $\closure{f} \leq \rho(g) \circ \closure{g \circ f}$ and
  $g \circ \closure{f} \leq \closure{g \circ f}$. 
\end{proof}

\begin{lemma}
  We have
  \begin{align}
    \label{eq:cotightTwo}
    \Rans{\closure{\rans{g} \circ f}}
    & = \rans{g} \circ \rans{f}\,.
  \end{align}
\end{lemma}
\begin{proof}
  By monotonicity, we have $\rans{g} \circ \rans{f}
  = \rans{(\rans{g} \circ f)}
  \leq \Rans{\closure{\rans{g} \circ f}}$.
  Let us show that the converse inclusion holds:
  \begin{align*}
    \Rans{\closure{\rans{g} \circ f}} \leq \rans{g} \circ \rans{f} &
    \Tiff \closure{\rans{g} \circ f} \leq \rani{(\rans{g} \circ
      \rans{f})} = \ranD{(\rans{g} \circ
      f)} \\
    & \Tiff \rans{g} \circ f \leq \ranD{(\rans{g} \circ f)} \,,
  \end{align*}
  where the latter relation holds, since $\ranD{k}$ is the least
  \cotight map above $k$.
\end{proof}

Consider the following operation on  the set $\LLi$:
\begin{align*}
  g \bullet f & \eqdef \closure{\rans{g} \circ f}\,.
\end{align*}
\begin{lemma}
  $\bullet$ is a semigroup operations on $\LLi$.
\end{lemma}
\begin{proof}
  Using the previous two lemmas, compute as follows:
  \begin{align*}
    h \bullet (g \bullet f) & = \closure{\rans{h} \circ \closure{\rans{g}
        \circ f}}  = \closure{\rans{h} \circ \rans{g} \circ f}\,,
    \tag*{by \eqref{eq:cotightOne}}
    \\
    (h \bullet g) \bullet f & = \closure{\Rans{\closure{\rans{h} \circ g}}
      \circ f}  = \closure{\rans{h} \circ \rans{g} \circ
      f}\,, \tag*{by \eqref{eq:cotightTwo}. \qquad \qedhere}
  \end{align*}
\end{proof}

For $f \in \LLi$, we define
\begin{align*}
  \Perp{f} & \eqdef \rani{\ladj{f}}  \,= \rho(\rans{f})\,.
\end{align*}
\begin{proposition}
  $(\LLi,\bullet)$ is a quantale and $\Perp{\fun}$ is a self-adjoint
  Serre Galois connection on $\LLi$.
\end{proposition}
\begin{proof}
  We firstly  argue that $\Perp{\fun}$ is a self-adjoint Galois connection
  on $\LLi$ which, moreover, satisfies the shift relations
  \eqref{eq:shift} w.r.t. the operation $\bullet$. We have
  \begin{align*}
    f \leq \Perp{g} = \rho(\rans{g}) \Tiff \rans{g} \leq \ladj{f} \Tiff g
    \leq \rani{\ladj{f}} = \rho(\rans{f}) = \Perp{f} 
  \end{align*}
  and
  \begin{align*}
    g \bullet f  = \closure{\rans{g} \circ f} \leq \rho(\rans{h}) = \Perp{h}
    & \Tiff \rans{g} \circ f \leq \rho(\rans{h})
    \\
    & \Tiff f \leq \rho(\rans{g}) \circ  \rho(\rans{h}) =
    \rho(\rans{h} \circ \rans{g}) \\
    & \Tiff  \rans{(\rans{h} \circ g)} =  \rans{h} \circ \rans{g} \leq \ladj{f}
    \\
    & \Tiff   \rans{h} \circ g  \leq \rani{\ladj{f}} 
    \\ 
    & \Tiff  h \bullet g = \closure{\rans{h} \circ g}  \leq
    \rani{\ladj{f}} = 
    \Perp{f}\,.
  \end{align*}
  Let now
  \begin{align*}
    g \lrimpl h & \eqdef \rho(\rans{g}) \circ h\,,
    & h \rlimpl f & \eqdef \Perp{h} \lrimpl \Perp{f}\,.
  \end{align*}
  It is immediate to see that $g \bullet f \leq h$ if and only if
  $f \leq g \lrimpl h$. To see that these two relations are equivalent to
  $g \leq h \rlimpl f$, use the shift relations.
\end{proof}

\begin{proposition}
  The nucleus induced by the self-adjoint Serre Galois connection
  $\Perp{\fun}$ is $\ranD{\fun}$. The structure
  $(\LLit,\bullet_{\ranD{}},\Perp{\fun})$ is that of a Girard quantale on the
  set of \cotight maps of $L$ and $\rans{\fun} : \LLit \rto \LLst$
  yields a Girard quantale isomorphism from
  $(\LLit,\bullet_{\ranD{}},\Perp{\fun})$ to $(\LLst,\circ,\Star{\fun})$.
\end{proposition}
\begin{proof}
  First we compute the nucleus induced by $\Perp{\fun}$:
  \begin{align*}
    \Perp{f}\Perp{}
    & = \rho(\rans{\rho(\rans{f})})
    = \rho(\ladj{\rani{\rans{f}}}) = \ranD{f}\,.
  \end{align*}
  Secondly, we observe that
  \begin{align*}
    g \bullet_{\ranD{}} f
    & = \ranD{\closure{\rans{g} \circ f}}
    = \rani{(\rans{g} \circ \rans{f})}\,,
  \end{align*}
  using equation \eqref{eq:cotightTwo}.
  Then, for $f,g$ \cotight, we have
  \begin{align*}
    \rans{(\Perp{f})} & = \rans{\rho(\rans{f})} = \Star{(\rans{f})}\,,
  \end{align*}
  and
  \begin{align*}
    \rans{(g \bullet_{\ranD{}} f)} & = \rand{(\rans{g} \circ \rans{f})} = \rans{g} \circ
    \rans{f}\,.
    \tag*{\qedhere}
  \end{align*}
\end{proof}

\begin{remark}
  Recall that $\LLi$ is isomorphic in the category of \SLatt to the
  tensor product $L^{op} \tensor L$.  The binary operation $\bullet$
  described here is, up to isomorphism, the structure described in
  \cite[Theorem 3.3.5]{KrumlPaseka2008}. Indeed, up to isomorphism,
  the elementary tensors are, for $x,y \in L$,
  \begin{align*}
    (y \tensor x)(t) & \eqdef
    \begin{cases}
      \top, & t = \top,\\
      y\,, & x \leq t < \top, \\
      \bot\,, & \text{otherwise}.
    \end{cases}
  \end{align*}
  Observing now that $\rans{(y \tensor x)} = c_{y} \circ a_{x}$, it is
  then immediate to verify that the formula used in
  \cite{KrumlPaseka2008} to define the semigroup structure holds:
  \begin{align*}
    (v \tensor u) \bullet (y \tensor x) 
    & = \closure{c_{v} \circ a_{u} \circ (y \tensor x)}
    =
    \begin{cases}
      \bot\,, & y \leq u\,, \\
      v \tensor x\,, & \text{otherwise}\,.
    \end{cases}
  \end{align*}
\end{remark}

\section{No unital extensions}

Let $L$ be a complete lattice which is not \cd.  In view of
Theorem~\ref{thm:TightIsUnital}, implying that the Girard quantale
$\LLt$ has no unit, it might be asked whether it is possible to embed
$\LLt$ into a unital Girard quantale. Without further constraints,
this is always possible, by firstly adding a unit, see
e.g. \cite[Lemma 1.1.11]{KrumlPaseka2008}, and then by embedding the
resulting quantale into its Chu construction.  Such an embedding,
however, does not preserve the duality.

\smallskip

We defined unitless Frobenius quantales as certain semigroups coming
with a notion of duality.  Given this choice of elementary operations,
a more appropriate demand is to find an embedding that preserve sups,
the binary multiplication, and the dualities.
Given that
\begin{align*}
  x \lrimpl y & = \Rneg{(\Lneg{y} \qmult x)}\,, & x \rlimpl y & =
  \Lneg{(y \qmult \Rneg{x})}\,, &
  \bigwedge_{i \in I} x_{i} & = \Rneg{(\bigvee_{i \in I} \Lneg{x_{i}})}
\end{align*}
such an embedding necessarily preserves the implications and the
infima. 
\begin{definition}
  We say that  a quantale embedding $\iota : (Q_{0},\qmult_{0}) \rto
  (Q_{1},\qmult_{1})$ is \emph{\strong} if it preserves infima and the
  two implications. 
\end{definition}

Theorem~\ref{thm:noexpansion} proves that such an embedding into a
unital quantale never exists, neither for $\LLt$, nor for any
Frobenius quantale with no unit.
\begin{theorem}
  \label{thm:noexpansion}
  If a Frobenius quantale has a \strong embedding into a unital
  Frobenius quantale, then it has a unit.
\end{theorem}
\begin{proof}
  Let $(Q,\qmult,\Lneg{\fun},\Rneg{\fun})$ be a Frobenius quantale and
  let $u \eqdef \bigwedge \set{x \lrimpl x \mid x \in Q}$. Then, for each
  $x \in Q$, $u \leq x \lrimpl x$ and so $x \qmult u \leq x$. Also,
  $u \leq \Lneg{x} \lrimpl \Lneg{x} = x \rlimpl x$, so
  $u \qmult x \leq x$.

  Let us suppose next that this quantale has a \strong embedding into a
  unital quantale $(Q',1,\qmult)$.  Let us regard $Q$ as a subset of
  $Q'$ closed under suprema, infima, the binary multiplication, and
  the two implications. Then, for each $x \in Q$, $1 \leq x \lrimpl x$
  and so $1 \leq u$. It follows that, for each $x \in Q$,
  $x = x \qmult 1 \leq x \qmult u$, thus $x = x \qmult u$. Similarly,
  $x \leq u \qmult x$ and $x = u \qmult x$.
\end{proof}

Notice in the previous proof that we might have $1 < u$. That is,
\strong embeddings need not to preserve units. Examples of \strong
embeddings among unital Girard quantales arise from chain refinements.
Let $C_{0},C_{1}$ be two complete chains with $C_{0} \subseteq C_{1}$
and $C_{0}$ being closed under meets and suprema of $C_{1}$. If
$C_{1} $ is finite, $C_{0}$ being closed under meets and suprema
simply means that the endpoints of $C_{1}$ belongs to $C_{0}$. Under
these assumptions, there is a \strong embedding of $\Homs{C_{0}}$ into
$\Homs{C_{1}}$ which does not preserve the unit, while it preserves
composition and the duality, see \cite[Proposition 81]{SAN-2021-JPAA}.

\section{Tigth endomaps of $M_{n}$}

As a last step, we investigate tight maps for some simple
non-distributive finite lattices.

\begin{lemma}
  \label{lemma:tightFromImage}
  Let $L$ be a finite lattice.
  Suppose $f : L \rto L$ is a \jp map such that $f(L)$ is
  distributive. Then $f$ is tight.  
\end{lemma}
\begin{proof}
  We use Lemma~\ref{lemma:generated} and show that
  $f = \bigvee c_{y_{i}} \circ a_{x_{i}}$ for some family
  $\set{(y_{i},x_{i}) \mid i \in I}$.

  That is, we are going to show that, given $x \in L$, there exists a
  family $\set{(y_{i},x_{i}) \mid i \in I}$ such that
  $c_{y_{i}} \circ a_{x_{i}} \leq f$, for each $i \in I$, and
  $\bigvee_{i \in I}c_{y_{i}} \circ a_{x_{i}}(z) = f(x)$.

  Write $f(x) = \bigvee_{i} f(y_{i})$ with $f(y_{i})$ join-prime in
  $f(L)$. For each $i \in I$, let $x_{i} \in L$ be maximal in the set
  $\set{z \mid f(y_{i}) \not \leq f(z)}$.

  Suppose $t \not \leq x_{i}$. Then $x_{i} < x_{i} \vee t$, so
  $f(y_{i}) \leq f(x_{i} \vee t) = f(x_{i}) \vee f(t)$, and since
  $f(y_{i}) \not \leq f(x_{i})$, then $f(y_{i}) \leq f(t)$.
  This shows that $c_{f(y_{i})}\circ a_{x_{i}} \leq f$.

  Observe that, for each $i \in I$, $x \not\leq x_{i} $, since
  otherwise $f(y_{i}) \leq f(x) \leq f(x_{i})$.
  Consequently
  \begin{align*}
    (\bigvee_{i \in I} c_{f(y_{i})}\cdot a_{x_{i}})(x) & = \bigvee_{i
      \in I} c_{f(y_{i})}(a_{x_{i}}(x)) = \bigvee_{i \in I} f(y_{i}) =
    f(x)\,.
    \tag*{\qedhere}
  \end{align*}

\end{proof}

\begin{figure}
  \centering
  \includegraphics[scale=0.8]{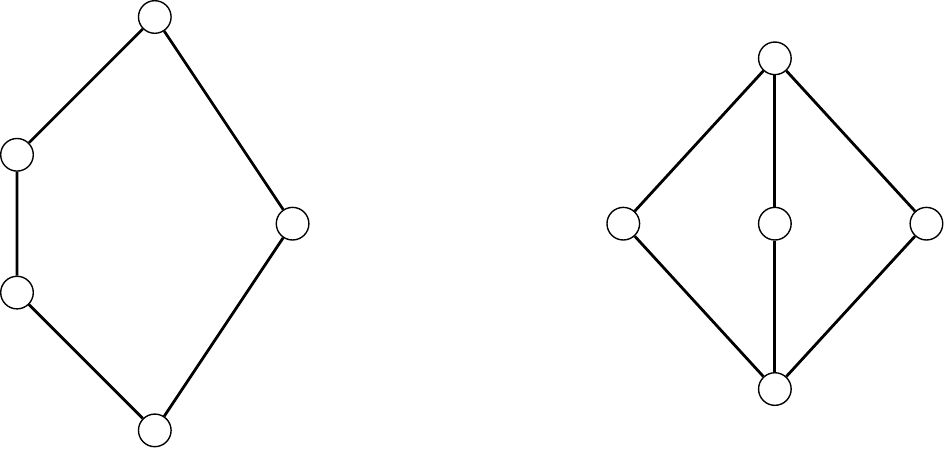}
  \caption{The pentagon $N_{5}$ (left) and the diamond $M_{3}$ (right).}
  \label{fig:pentagonDiamond}
\end{figure}

\begin{proposition}
  If $L$ is either the pentagon $N_{5}$ or the diamond lattice $M_{3}$
  (see Figure~\ref{fig:pentagonDiamond}), then
  $\LLt = \LLs \setminus \HLL^{\times}$, where $\HLL^{\times}$ is the
  set of order isomorphisms of $L$.
\end{proposition}
\begin{proof}
  If $f \in \HLL^{\times} \cap \LLt$, then
  $id_{L} = f^{-1} \circ f \in \LLt$, thus $L$ is
  distributive by Raney theorem. Therefore, since neither $N_5$ nor
  $M_3$ are distributive, if $L$ is any of these two lattices, 
  then
  $\HLL^{\times} \cap \LLt = \emptyset$.
  On the other hand, if
  $L \in \set{N_{5},M_{3}}$ and $f \not \in \HLL^{\times}$, then $f$
  is not injective and its image has cardinality at most $4$. Since
  every lattice of cardinality at most $4$ is distributive, then $f$
  is tight, using Lemma~\ref{lemma:tightFromImage}.
\end{proof}

Notice that, for $L = N_{5}$, $\LLs^{\times} = \set{\id_{L}}$, while
for $L = M_{3}$, $\LLs^{\times}$ is isomorphic to the permutation
group on $3$ elements.

The relation $\LLt = \LLs \setminus \LLs^{\times}$ does not hold in
general. We exemplify this by characterizing $\LLt$ for $L = M_{n}$,
the modular lattice of height $3$ $n$ atoms. The cadinalities of
$\Homs{M_{n}}$ have been described in \cite{Valencia2020}. To describe
$\Homs{M_{n}}^{t}$, let $At(M_{n})$ denote the set of atoms of
$M_{n}$.

\begin{proposition}
  \label{prop:tightMn}
  A \jp map $f: M_{n} \rto M_{n}$ is tight if and only if its image
  $f(M_{n})$ is distributive, if and only if the cardinality of
  $f(M_{n}) \cap At(M_{n})$ is at most $2$.
\end{proposition}
\begin{proof}
  If $\card{f(M_{n}) \cap At(M_{n})} \leq 2$, then
  $\card{f(M_{n})} \leq 4$ and $f(M_{n})$ is a distributive lattice.
  If $\card{f(M_{n}) \cap At(M_{n})} \geq 3$, then $f(M_{n})$ is
  isomorphic to $M_{k}$ with $3 \leq k \leq n$, thus it is not
  distributive.

  By Lemma~\ref{lemma:tightFromImage}, if $f(M)$ is distributive, then
  $f$ is tight.  Conversely, suppose that
  $\card{f(M_{n}) \cap At(M_{n})} \geq 3$. Let $a,b,c$ be such that
  $f(a),f(b),f(c) \in At(M_{n})$ and observe that $f(a) \neq
  \bot$. Then
  \begin{align*}
    \rand{f}(a)
    & = \bigvee_{a \not\leq y} \bigwedge_{t \not\leq y} f(t) 
    = \bigvee_{y \in At(M_{n}), y \neq a} \;\;\bigwedge_{t \in At(M_{n}), t \neq y} f(t) \\
    & = \bigvee_{y \in At(M_{n}), y \neq a} f(a) \land \bigwedge_{t
      \in At(M_{n}), t \neq y,a} f(t) \\
    & = \bigvee_{y \in At(M_{n}), y \neq a} \bot = \bot\,,
  \end{align*}
  so $f$ is not tight.
\end{proof}

\begin{corollary}
  We have
  \begin{align*}
    \card{\Homs{M_{n}}^{t}} & = 2 + 2n + 2n^{2} +
    \binom{n}{2}n(n-1) 
    = \frac{1}{2}n^{4} - n^{3} + \frac{5}{2}n^{2} + 2 n + 2\,.
  \end{align*}
\end{corollary}
\begin{proof}
  If a tight function is not of the form $c_{\bot},c_{\top}$,
  $c_{j},a_{j}$ (for $j \in At(M_{n})$), $c_{j} \circ a_{m}$,
  $c_{j} \vee a_{m}$ ($j,m \in At(M_{n})$), then it takes two atoms
  and sends it to distinct atoms, while all other atoms are sent to
  $\top$.
\end{proof}

Observe that if $f$ is tight and $\card{f(M_{n})} \leq 3$, then $f$
has one of this form:
\begin{align*}
  c_{y} \circ a_{x}\,, & \qquad \quad c_{y} \vee a_{x}\,
\end{align*}
where $x,y \in M_{n}$. Notice that $c_{y} \vee a_{x}$ is the map which
sends $t$ such that $\bot < t \leq x$ to $y$ and $t$ such that
$t \not\leq x$ to $\top$.
If  $\card{f(M_{n})} = 4$, then $f$ is of the form
\begin{align*}
  f_{x_{1},y_{1},x_{2},y_{2}} & \eqdef c_{y_{2}}\circ a_{x_{1}} \vee
  c_{y_{1}}\circ a_{x_{2}}\,,
\end{align*}
where $(x_{1},x_{2})$ and $(y_{1},y_{2})$ are pairs of distinct
atoms of $M_{n}$. Observe that
\begin{align}
  \label{eq:Gen}
  f_{x_{1},y_{1},x_{2},y_{2}}(t) & =
  \begin{cases}
    \bot \,, & t = \bot\,, \\
    y_{1}\,, & t = x_{1}\,, \\
    y_{2}\,, & t = x_{2}\,, \\
    \top \,, & \text{otherwise}\,.
  \end{cases}
\end{align}

The following proposition achieves the goal of giving a more immediate
formula for the negation in the quantale $(\Homs{M_{n}}^{t},\circ)$.
\begin{proposition}
  We have
  \begin{align*}
    \Star{(c_{y} \circ a_{x})}
    & = c_{x} \vee a_{y} \Tand
    \Star{(f_{x_{1},y_{1},x_{2},y_{2}})} = f_{y_{1},x_{2},y_{2},x_{1}}
  \end{align*}
  where above $x_{1} \neq x_{2}$ and $y_{1} \neq y_{2}$.
\end{proposition}
\begin{proof}
  Let $\alpha_{x}$ be the \mp endomap such that $\alpha_{x}(t) = \top$
  if $x \leq t$, and $\alpha_{x}(t) = \bot$ otherwise. 
  It is easily verified that $\radj{c_{y}} = \alpha_{y}$ and that
  $\rans{(\alpha_{y})} = a_{y}$, so $\Star{(c_{y})} = a_{y}$. From
  this it also follows that $\Star{(a_{x})} = c_{x}$.  Observe now 
  that
  $c_{y} \circ a_{x} = c_{y} \land a_{x}$, from which it follows that
  \begin{align*}
    \Star{(c_{y} \circ a_{x})} & = \Star{(c_{y} \land a_{x})} =
    \Star{(c_{y})} \vee \Star{(a_{x})} = a_{y} \vee c_{x}\,.
  \end{align*}

  Using this and assuming that $x_{1},x_{2}$ and $y_{1},y_{2}$ are
  pairs of distinct atoms, we also compute as follows:
  \begin{align*}
    \Star{(f_{x_{1},y_{1},x_{2},y_{2}})} & = \Star{((c_{y_{2}}\circ
      a_{x_{1}}) \vee (c_{y_{1}}\circ a_{x_{2}}))} \\
    &
    = \Star{(c_{y_{2}} \circ a_{x_{1}})} \land \Star{(c_{y_{1}}
      \circ a_{x_{2}})} \\
    &
    = (a_{y_{2}}\vee
    c_{x_{1}}) \land (a_{y_{1}}\vee c_{x_{2}}) = f_{y_{1},x_{2},y_{2},x_{1}}\,,
  \end{align*}
  where the last equality is derived by observing that the pointwise
  meet of $ c_{x_{1}}\vee a_{y_{2}}$ and $c_{x_{2}} \vee a_{y_{1}}$
  satisfies the pattern for $f_{y_{1},x_{2},y_{2},x_{1}}$ given in
  equation~\eqref{eq:Gen}. 
  For example, if $t = y_{1}$, then
  $(a_{y_{2}}\vee c_{x_{1}})(t) = \top$ and
  $(a_{y_{1}}\vee c_{x_{2}})(t) = x_{2}$. A similar argument is used
  when $t = y_{2}$.  Finally, if
  $t\not\leq y_{i}$, $i = 1,2$, then
  $(a_{y_{2}}\vee c_{x_{1}})(t) = (a_{y_{1}}\vee c_{x_{2}})(t) =
  \top$.
\end{proof}

We might further study quantales of the form $\LLt$ by trying to
identify similarities and differences with unital quantales.  Let us say that an
element $p \in Q$ is \emph{positive} if, for all $x \in Q$,
\begin{align*}
  x & \leq x \qmult p \land p \qmult x\,.
\end{align*}
Let us call a quantale  \emph{\positive}
if each element of the form $x \lrimpl x$ or $x \rlimpl x$ is
positive.
The proof Theorem~\ref{thm:noexpansion} immediately yields the
following statement.
\begin{lemma}
  A Frobenius quantale is unital if and only if it is positive and
  positive elements are closed under infima.
\end{lemma}
Similarly, a residuated partially-ordered semigroup (see \cite[\S
3.2]{GJKO}) is positive if each element of the form $x \lrimpl x$ or
$x \rlimpl x$ is positive.  It was argued in \cite[Theorem
3.27]{Blount1999} that a residuated partially-ordered semigroup embeds
into a residuated partially-ordered monoid if and only if it is
\positive.

The following statement shows that positivity is no longer sufficient
when we step from residuated semigroups to Frobenius quantales and
residuated lattices---where for residuated lattices, a straightforward
modification of Theorem~\ref{thm:noexpansion} shows that the finite
unitless quantales $\Homs{M_{n}}^{t}$ cannot be embedded into unital
ones.
\begin{proposition}
  The quantales $\Homs{M_{n}}^{t}$ are \positive.
\end{proposition}
\begin{proof}
  Let us show that $f \lrimpl f$ is above the identity whenever
  $f \in \Homs{M_{n}}^{t}$.  It is easily seen that, within
  $\Homs{M_{n}}^{t}$, $f \lrimpl g = \rand{(\rho(f) \circ g)}$.

  Let $h = \rho(f) \circ f$ and recall that $id_{M_{n}} \leq h$. It
  follows that, for each $a \in At(M_{n})$,
  \begin{align*}
    h(a) & =
    \begin{cases}
      \top\,,& f(a) = f(\top)\,, \\
      a\,, & f(a) < f(\top)\,.
    \end{cases}
  \end{align*}
  The restriction of $f$ to the set of fixed points of $h$ yields an
  order isomorphism with the image of $f$. Therefore, there can be at
  most two atoms that are fixed by $h$.  The map that is equal to $h$
  except that it sends $\bot$ to $\bot$ is then join-preserving and
  tight, by Proposition~\ref{prop:tightMn}. Then
  $f \lrimpl f = \rand{h}$ is this map, and clearly this map is above
  the identity.
\end{proof}

Therefore, a main difference with unital quantales is that positive
elements are not closed under infima.
This can also be directly observed as follows. Consider idempotent
positive elements of $\HLL$ as \jp closure operators on $L$. The set
of fixed points of such a \jp closure operator is closed under infima
(as usual) and, moreover, under suprema. This establishes a dual
bijection between \jp closure operators on $M_{n}$ and sublattices of
$M_{n}$, see \cite[\S 3]{SAN-2019-WORDS}. 
Then, by Proposition~\ref{prop:tightMn}, tight closure operators
correspond to distributive sublattices. As suggested in
Figure~\ref{fig:jDistrLatt}, the join of two distributive sublattices
of $M_{n}$ might not be distributive. Dually, this amounts to saying
that the meet in $\Homs{M_{n}}$ of tight closure operators might not
be tight; then, Figure~\ref{fig:jDistrLatt} tells us that the meet in
$\Homs{M_{n}}$ of two tight closure operators $j_{1},j_{2}$ is the
identity $\id_{M_{n}}$. Thus, the meet $j_{1} \land j_{2}$ in
$\Homs{M_{n}}^{t}$, computed as the interior of the identity
$\rand{(\id_{M_{n}})}$, is easily seen to be the bottom of the
quantale. Next, the bottom of a quantale is never positive if the
quantale has more than one element.
\begin{figure}
  \centering
  \includegraphics{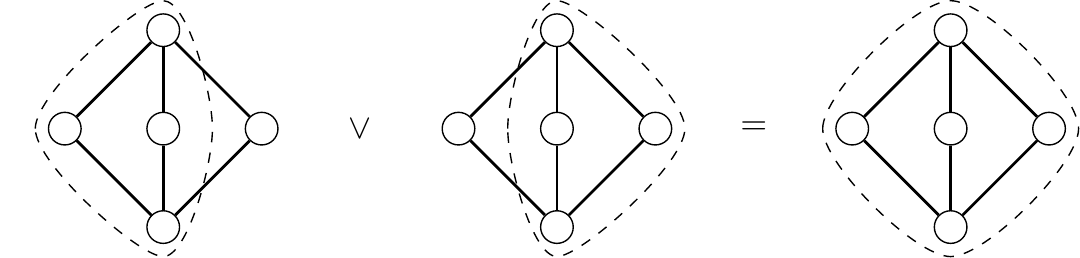}
  \caption{Join of distributive lattices as meet of tight closure
    operators}
  \label{fig:jDistrLatt}
\end{figure}

\bigskip

\ifacs\noindent {\bf Acknowledgment.}~\else
\paragraph{\bf Acknowledgment}\fi The authors are thankful to Nick Galatos
for precious pointers and remarks.

\ifacs
\backmatter

\else
\bibliographystyle{abbrv}
\bibliography{biblio}
\fi

\end{document}